\documentclass[a4paper,10pt]{article}
\usepackage{amsthm}
\newtheorem{lemma}{Lemma}

\newtheorem{cor}{Corollary}
\newtheorem{theorem}{Theorem}
\usepackage{graphicx}
\usepackage{graphicx,amssymb}
\usepackage{hyperref}
\usepackage{amsmath}
\usepackage{epsfig, color, graphicx}
\usepackage{enumerate}
\usepackage[margin=1.0in]{geometry}
\usepackage{amssymb}

\title{Point visibility graph recognition is NP-hard 
\thanks  
{ A part of the work was done when the author visited Carleton University under
DFAIT Commonwealth Scholarship of the Government of Canada. 
}
}
\author{
}
\begin{document}

\date{\vspace{-5ex}}
\maketitle

\begin{center}
 \mbox{\begin{minipage} [b] {3in}
 
\centerline{Bodhayan Roy}
\vspace{5mm}
\centerline{ School of Technology and Computer Science}
\centerline{ Tata Institute of Fundamental Research}
\centerline{Mumbai 400005, India}
\centerline{bodhayan@tifr.res.in}
\end{minipage}}
\end{center}
\begin{abstract}
Given a 3-SAT formula, a graph can be constructed in polynomial time such that the graph is a point visibility
graph if and only if the 3-SAT formula is satisfiable.
This reduction establishes that the problem of recognition of point visibility graphs is NP-hard.
\end{abstract}
\section{Introduction}
 The visibility graph is a fundamental structure studied in the field of computational geometry  
and geometric graph theory \cite{bcko-cgaa-08, g-vap-07}. 
 Some of the early applications of visibility graphs included computing 
Euclidean shortest paths in the presence of obstacles \cite{lw-apcf-79} and decomposing 
two-dimensional shapes into clusters \cite{sh-ddsg-79}. 
Here, we consider problems from visibility graph theory.
$\\$  $\\$
Let $P=\{ p_1, p_2, ..., p_n \}$ be a set of $n$ points in the plane. 
We say that two points $p_i$ and $p_j$ of $P$ 
are \emph{visible} to each other if the line 
segment $p_ip_j$ does not contain any other point of $P$. In other words,
$p_i$  and $p_j$ are visible to each other if $P \cap  {p_ip_j}=\{p_i, p_j\}$. If two points are not visible, 
they are called  \emph{ invisible } to each other.  
If a point $p_k \in P$ lies on the segment $p_ip_j$ connecting two points $p_i$ and $p_j$ in $P$, 
we say that $p_k$ blocks the visibility between $p_i$ and $p_j$, and
$p_k$ is called a \emph{ blocker} in $P$.  
$\\$ $\\$
The 
\emph{ point visibility graph} (denoted as PVG)
 of $P$ is defined 
by associating a vertex $v_i$ with each point $p_i$ of $P$ 
and
an undirected edge 
$(v_i, v_j)$
of the PVG if  $p_i$ and $p_j$
are visible to each other. Observe that if no three points of 
$P$ are collinear, then the PVG is a complete graph as
each pair of points in $P$ is visible since there is no blocker in $P$. 
Point visibility graphs have been studied in the context of connectivity \cite{viscon-wood-2012},
chromatic number and clique number \cite{kpw-ocnv-2005, p-vgps-2008}.
For review and open problems on point visibility graphs, see 
Ghosh and Goswami \cite{prob-ghosh}.
$\\ \\$
Given a point set $P$, the PVG of $P$ can be computed in polynomial time. Using the result of
 Chazelle et al. \cite{cgl-pgd-85} or Edelsbrunner et al. 
\cite{Edelsbrunner:1986:CAL},  this can be achieved in $O(n^2)$  
time.
Consider the opposite problem:  given a graph $G$, determine if there is a set of points $P$ 
whose point visibility graph is $G$. This problem is called the point visibility graph 
\emph{recognition} problem \cite{prob-ghosh}.
Identifying the set of properties satisfied by all visibility graphs is called the 
point visibility graph \emph{characterization} problem.
The problem of 
actually drawing one such set of points $P$ whose point visibility graph is the given graph $G$, 
is called the point visibility graph \emph{reconstruction} problem. Such a point set itself is called a 
\emph{visibility embedding} of $G$.
$\\ \\$
Ghosh and Roy \cite{recogpvg-2014}
presented a complete characterization for 
planar point visibility graphs,
 which leads to a linear time recognition and reconstruction algorithm.
For recognizing arbitrary point visibility graphs, they 
presented three necessary conditions, and 
gave a polynomial time algorithm for testing the first necessary condition.
However, it is not clear whether the other two necessary conditions can be checked in polynomial time.
If a 
set of necessary and sufficient conditions for recognizing point visibility graphs can be found such that they can be tested 
in polynomial time, then the recognition problem
lies in P. 
So, it is necessary to investigate the complexity issues of recognizing point visibility graphs.
This problem is known to be in PSPACE,
which is the only upper bound known on the complexity of the problem \cite{recogpvg-2014, prob-ghosh}. On the other
hand, problems of minimum vertex cover, maximum independent set, and maximum
clique of point visibility graphs are shown to be NP-hard   \cite{recogpvg-2014, prob-ghosh}.
%
%
%
%
$\\ \\$
In this paper,
we show that the recognition problem for $PVGs$ is NP-hard.
In Section \ref{secsgg}, we develop a \emph{slanted grid graph} (denoted as $SGG$) that has a unique visibility embedding. 
The embedding of the $SGG$ contains a gridlike structure.
In Section \ref{secmsgg}, we transform the slanted grid graph into a \emph{modified slanted grid graph} (denoted as $MSGG$), 
that also has a unique visibility embedding. 
The unique visibility embedding of the $MSGG$ also contains a gridlike structure, however, an area inside the grid is devoid of 
points. This area is later used to embed another graph inside the $MSGG$. 
In Section \ref{consmsgg} we describe the construction of the $MSGG$.
In Section \ref{unqemmsgg} we begin with lemmas on some less complex graphs and finally prove
that the $MSGG$ has a unique visibility embedding.
In Section \ref{sec3sat} we first introduce a \emph{3-SAT graph}, that has vertices and edges corresponding
to a given 3-SAT formula and its size   polynomial in the 
size of the given 3-SAT formula. 
We describe the construction of the 3-SAT graph in Section \ref{consalg}.
In Section \ref{subsecconsred}, we strategically add this graph to a large enough $MSGG$, and 
call the result a \emph{reduction graph}. The reduction graph inherits collinearity conditions from 
the $MSGG$, so that if it has a visibility embedding, then the configuration of its points belonging
to the  3-SAT graph corresponds   to a truth assignment of the given 3-SAT formula. In Section \ref{subsecconsred},
we prove that if the given 3-SAT formula is not satisfiable, then the reduction graph has no visibility embedding.
In Section \ref{3satred}, we prove the converse direction of the reduction, i.e., that if 
the given 3-SAT formula is satisfiable, then the reduction graph has a visibility embedding. This completes
the reduction.
In Section \ref{conclrem}, we conclude the paper with a few remarks.

\section{Slanted grid graphs} \label{secsgg}
In this section, we define 
a special type of $PVG$ called the
\emph{slanted grid graph} ($SGG$). Intuitively, an $SGG$ is a $PVG$
resembling a grid graph \cite{diestel} with two extra vertices so that in its visibility embedding, every line passes through
at least one of these two vertices. These two extra vertices are called \emph{vertices of convergence}.
$\\ \\ $
Let $G=(V,E)$ be 
the $PVG$ of a point set $P$.
Let $f:V \longrightarrow P$ be a bijection.
We say that the pair $\langle P,f \rangle$ is a \emph{visibility embedding} of $G$
if
$$P \cap p_ip_j = \{p_i,p_j\} \Longleftrightarrow (f^{-1}(p_i),f^{-1}(p_j))\in E \enspace.$$
Let $G_0=(V_0,E_0)$ be a PVG, and 
$\xi=\langle\{p_1,p_2,...,p_n\}, f\rangle$ and $\xi'=\langle\{p_1',p_2',...,p_n'\}, f'\rangle$ be two visibility embeddings of $G_0$.
A line passing through some points   of $\xi$ is simply referred to as a \emph{line} in $\xi$. 
Let $L$ be a line in $\xi$
and let $\langle p_{i_1}, p_{i_2},...,p_{i_{\ell}}\rangle$
be the sequence of all points in $\xi$ that lie on $L$ in this order.
We say that $L$ is \emph{preserved} in $\xi'$
if all the points in the sequence $L'=\langle f'(f^{-1}(p_{i_1})), f'(f^{-1}(p_{i_2})),...,f'(f^{-1}(p_{i_{\ell}}))\rangle$
lie on the same line, in the same order and no other point of $\xi'$ lies on $L'$.
%
 
$ \\ $
Let $m\in \mathbb{N}$ and $n\in \mathbb{N}$ be two numbers such that   $m \geq 3$ and $n\geq 3$.
Consider graph $G_0 = (V_0,E_0)$, where $V_0$ and $E_0$
are defined as follows.
\begin{eqnarray*}
V_0 &=& \{v_{i,j} \vert 1\leq i \leq n \textrm{ and }1\leq j \leq m\} \cup \{v_1,v_2\}\\
E_0 &=& \{(v_{i,j},v_{k,l}) \vert i\neq k \textrm{ and } j\neq l\} 
\cup \{(v_1,v_{1,j})\vert 1\leq j\leq m\} \cup \{(v_2,v_{i,1})\vert 1\leq i\leq n\}\cup\{(v_1,v_2)\}\\
 &\cup&  \{(v_{i,j},v_{i+1,j}) \vert 1 \leq i \leq n-1 \textrm{ and } 1 \leq j \leq m  \} 
 \cup \{(v_{i,j},v_{i,j+1}) \vert 1 \leq i \leq n  \textrm{ and } 1 \leq j \leq m-1  \} 
\end{eqnarray*}
  Observe that $\vert V_0 \vert = mn +2 $    and $\vert E_0 \vert = (mn - m - n + 3)mn + 1$.   
We call this graph a \emph{slanted grid graph} (SGG),
which is a $PVG$ with visibility embedding as explained below.
Consider a set of points 
$$P = \{p_{i,j} \vert 1\leq i \leq n \textrm{ and }1\leq j \leq m\} \cup \{p_1,p_2\}$$
and associate $v_1$ to $p_1$, $v_2$ to $p_2$ and $v_{i,j}$ to $p_{i,j}$ for $1\leq i\leq n$ and $1\leq j\leq m$.
A line passing through exactly two embedding points of $P$ is called an
 \emph{ordinary line}. A line passing through three or more embedding points of $P$ is called a
 \emph{non-ordinary line}.
Choose the coordinates of the points in $P$
in such a way that the non-ordinary lines in $P$ contain   
$\langle p_1,p_{1,j},p_{2,j},...,p_ {n,j}\rangle$ for $1\leq j\leq m$
and $\langle p_2,p_{i,1},p_{i,2},...,p_ {i,m}\rangle$ for $1\leq i\leq n$
(Figure \ref{sgg}).
Then $P$ is a visibility embedding of $G_0$, and points of $P$ are referred to as \emph{embedding points}.

 $ \\ $
In 
the following lemma,
we prove that this visibility embedding is actually unique,
up to the preservation of the lines.
\begin{figure}[h] 
\begin{center}
\centerline{\hbox{\psfig{figure=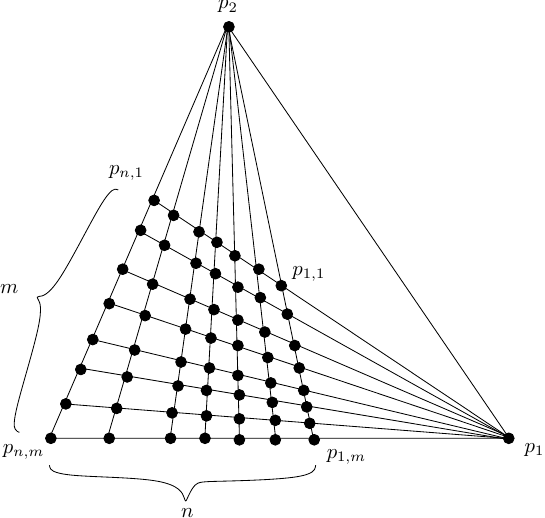,width=0.700\hsize}}}
\caption{Visibility embedding of a slanted grid graph.   The graph contains $mn +2$ vertices and 
$(mn-m-n+3)mn + 1$ edges.}
\label{sgg}
\end{center}
\end{figure}
\begin {lemma} \label{lem1}
$G_0$ has a unique visibility embedding, up to the preservation of lines (Figure \ref{sgg}).
\end {lemma}
\begin{proof} 
%
Suppose that the  embedding points $\{ p_{1,1}, p_{2,1}, \ldots, p_{n,1} \}$ lie on both sides of $\overleftarrow{p_2p_1}$.
Let $p_{x,1}$ and $p_{y,1}$ be the  embedding points from $\{ p_{1,1}, p_{2,1}, \ldots, p_{n,1} \}$
such that $\overleftarrow{p_2p_{x,1}}$ and $\overleftarrow{p_2p_{y,1}}$ make the smallest angles with $\overleftarrow{p_2p_1}$
to its left and right respectively. Now, only one  embedding point among $p_{x,1}$ and $p_{y,1}$ (say, $p_{x,1}$) can be $p_{1,1}$.   
So, some point $p_{a,b}$ must block $p_{y,1}$ from $p_1$. But then $\overleftarrow{p_2p_{a,b}}$ forms a smaller angle 
than $\overleftarrow{p_2p_{y,1}}$ with $\overleftarrow{p_2p_1}$ on the same side, a contradiction. 
Therefore, embedding points of $\{ p_{1,1}, p_{2,1}, \ldots, p_{n,1} \}$, and the remaining embedding points of
$P$ must lie on the same side of $\overleftarrow{p_2p_1}$.
$\\ \\$
Since $p_1$ and $p_2$ are mutually visible, no other embedding point can lie on $\overline{p_1p_2}$.
Consider an embedding point of $\{ p_{1,1}, p_{2,1}, \ldots, p_{n,1} \}$ lying on $\overleftarrow{p_1p_2}$.
Since the consecutive  embedding points of $\{p_1, p_{1,1}, p_{2,1}, \ldots, p_{n,1} \}$ are all visible to each other 
and all of them are visible from $p_2$, and $n \geq 3$, they must either be collinear or form a reflex chain facing $p_2$.
In either of these two cases, none of the embedding points of $\{ p_{1,1}, p_{2,1}, \ldots, p_{n,1} \}$ lies on $\overleftrightarrow{p_1p_2}$.
Since $m \geq 3$, an analogous reasoning shows that none of the embedding points of 
 $\{ p_{1,1}, p_{1,2}, \ldots, p_{1,m} \}$ lies on $\overleftrightarrow{p_1p_2}$.
 Hence, no remaining point of $P$ lies on $\overleftrightarrow{p_1p_2}$. 
$\\ \\$
There are $m+1$ and $n+1$ vertices in $G_0$ adjacent to $v_1$ and $v_2$  respectively. 
So there are $m+1$ and $n+1$ rays originating from $p_1$ and $p_2$ in a visibility embedding of $G_0$ respectively.
Consider the line passing through $p_1$ and $p_2$. 
We
leave aside the two rays from $p_1$ and $p_2$ because we have proved that no remaining point of $P$ lies on $\overleftrightarrow{p_1p_2}$.
   Observe that embedding points can only be placed on the possible intersection points of the remaining rays.
There are at most $m \times n$ intersection points formed by the remaining rays. 
However, since there are $m \times n$ remaining vertices of $G_0$, there are exactly 
$m \times n$ intersection points formed by the remaining rays.  
We call this configuration a \emph{slanted grid}.
  Observe that neither $p_1$ nor $p_2$ can be placed inside the convex hull of the other points,
 because otherwise 
 either not all embedding points lie on the same side of $\overleftarrow{p_2p_1}$, or some other embedding point 
 lies on $\overleftrightarrow{p_1p_2}$, a contradiction.
Wlog we assume that $p_1$ is placed
to the right of all other points, and $p_2$ is placed above all the other points (see Figure \ref{sgg}).
For convenience, we refer to the rays from $p_1$ as \emph{horizontal} and the rays from $p_2$ 
as \emph{vertical}.
Hence, the embedding points can only permute their positions
on the intersection points of rays.

$ \\$
Since $\{v_{1,1}, v_{2,1}, \ldots, v_{n,1} \}$ are adjacent to $v_2$, 
the embedding points $\{p_{1,1}, p_{2,1}, \ldots, p_{n,1} \}$ must occur on the topmost horizontal ray from $p_1$.
Since $v_{1,1}$ is adjacent to $v_1$, $p_{1,1}$ must be embedded to the left of $p_1$ with no other embedding point 
on $p_{1,1}p_1$. For $1 \leq i \leq n-1$,  $v_{i+1,1}$ is adjacent to $v_{i,1}$, and therefore,
$p_{i+1,1}$ must be embedded to the left of $p_{i,1}$ with no other embedding point 
on $p_{i+1,1}p_{i,1}$. Hence, the topmost horizontal line is preserved.
Since the vertices $\{v_{1,2}, v_{2,2}, \ldots, v_{n,2} \}$ are the only vertices other than $v_2$
adjacent to all the vertices $\{v_{1,1}, v_{2,1}, \ldots, v_{n,1} \}$,
  $\{p_{1,2}, p_{2,2}, \ldots, p_{n,2} \}$ must occur on the second-topmost horizontal ray from $p_1$.
By applying the previous reasoning, the second-topmost horizontal line is also preserved.
Similar arguments hold for other horizontal and vertical rays.
\end{proof}
 \begin{figure}[h] 
\begin{center} 
\centerline{\hbox{\psfig{figure=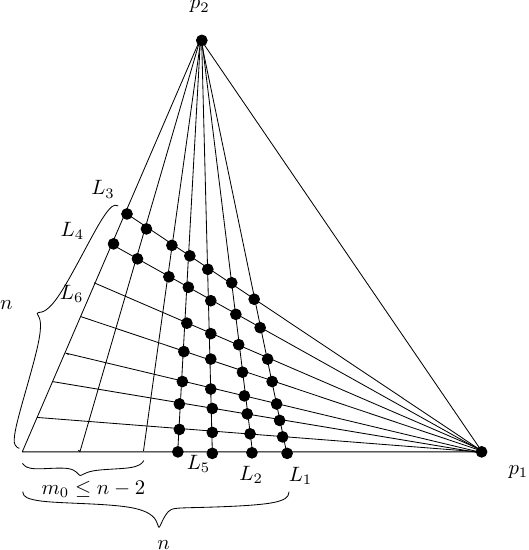,width=0.700\hsize}}}
\caption{An embedding of the $SGG$ after deletion of vertices.}
\label{dsgg}
\end{center}
\end{figure}
\begin{figure}[h] 
\begin{center} 
\centerline{\hbox{\psfig{figure=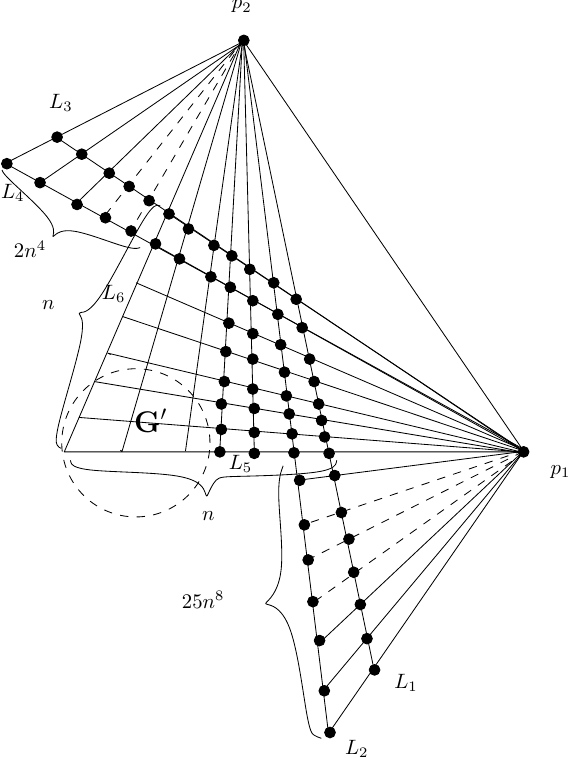,width=0.700\hsize}}}
\caption{Visibility embedding of a modified slanted grid graph. $G'$ is not a part of the MSGG.  }
\label{msgg}
\end{center}
\end{figure}
\section{Modified slanted grid graphs} \label{secmsgg}
In a visibility embedding of an $SGG$, every intersection point contains an embedding point. However,
if we delete some embedding points   (Figure \ref{dsgg}), then it is not clear whether the visibility graph of the remaining point set 
has a unique visibility embedding.
In order to ensure a   unique embedding after deletion, some new vertices are added to $G_0$ which facilitate 
a unique visibility embedding of the resulting graph $G$. In general, the vertices are added such that in every embedding there 
are four lines passing through a large number of embedding points, that enforce certain collinearity conditions 
and result in a unique visibility embedding up to the preservation of lines.
\subsection{Construction of a modified slanted grid graph } \label{consmsgg}
 Consider the unique embedding $ \xi$ of a slanted grid graph  $G_0$ with $n \times n$ embedding points
 and the two embedding points of convergence.
 We construct a modified slanted grid graph (denoted as $MSGG$) $G(V,E)$
 by adding some vertices to $G_0$ and then deleting some other vertices from it.
 We now describe the modification of $G_0$. Let $G_0$ be the $SGG$ on $(n \times n) + 2$ vertices defined in the previous section.
 Consider a positive intger $m_0$,   $m_0 \leq n-2$.
 Note that the removal of vertices from $V$ also implies the removal of their incident edges from $E$.
 We make the following modifications in $G_0$ to construct $G$.
 \begin{eqnarray*}
V &=& (V_0 \setminus \{ v_{i,j} \vert   n-m_0+1 \leq i \leq n  \textrm{ and }     3\leq j \leq n      \} )  
\cup \{v_{i,j} \vert n+1 \leq i \leq 2n^4-n \textrm{ and } 1\leq j \leq 2\}  \\
& &  \cup \{v_{i,j} \vert 1\leq i \leq 2 \textrm{ and } n+1 \leq j \leq 25n^8 + n\}                 \\
E &=& E_0 \cup \{ (v_{i,j},v_{i-1,j}) \vert  n+1 \leq i \leq 2n^4 +n \textrm{ and } 1 \leq j \leq 2 \}\\
& & \cup \{  (v_{i,1}, v_{j,2}) \vert  1 \leq i, j \leq 2n^4 +n       \} \\
& & \cup \{  (v_{i,j}, v_{k,l}) \vert  n+1 \leq i \leq 2n^4 + n \textrm{ and } 1 \leq j \leq 2 
\textrm{ and }  1 \leq k \leq n-m_0          \textrm{ and }         3 \leq l \leq n           \} \\ 
& & \cup \{ (v_{i,j-1},v_{i,j}) \vert   1 \leq i \leq 2  \textrm{ and } n + 1 \leq j \leq 25n^8 + n \} \\
& & \cup \{  (v_{1,i}, v_{2,j}) \vert  1 \leq i, j \leq 25n^8 + n      \} \\
& & \cup \{  (v_{i,j}, v_{k,l}) \vert  1 \leq i \leq 2 \textrm{ and }  n+1 \leq j  \leq 25n^8 +n
\textrm{ and }  3 \leq k \leq n-m_0           \textrm{ and }         1 \leq l \leq n           \} \\ 
& & \cup \{  (v_{i,j}, v_{k,l}) \vert  1 \leq j,k \leq 2 \textrm{ and }  3 \leq i   \leq 2n^4 + n
\textrm{ and }    3 \leq l \leq 25n^8 + n          \} \\ 
\end{eqnarray*}
Intuitively the $m_0(n-2)$ vertices that are deleted correspond to an $m_0 \times (n-2)$   grid in the visibility embedding of $G$.
  Since $m_0 \leq n-2$, the vertices of $L_1$ and $L_2$ are not deleted. 
  If $m_0^2$ vertices are later added to the graph with suitable adjacency relationships,
then they can be forced to be embedded on particular horizontal lines (see the proof of Lemma \ref{canemb}).
$\\ \\$
Now we construct a visibility embedding of the modified $G$, from the initial unique visibility
embedding of $G$ in Figure \ref{sgg}.
 Let $L_1$ and $L_2$ be the rightmost and second-rightmost lines of the visibility-embedding of an $MSGG$ (Figure \ref{msgg}).
 Let $L_3$ and $L_4$ be the topmost and second-topmost lines of the visibility-embedding of an $MSGG$.
 The
 bottommost horizontal line and leftmost vertical line are labelled $L_5$ and $L_6$, respectively.
As shown in Figure \ref{msgg}, the two points of convergence are above and to the 
right of the embedding. As before, $p_{i,j}$ is the embedding point corresponding to the vertex $v_{i,j}$.
%
%
\begin{enumerate}
 
\item \label{secstep} Delete the $(n-2) \times m_0$ bottom-left subgrid of $G$ (See Figure \ref{dsgg}).
 At a later stage, we embed a gadget in the space thus created.
\item \label{thrstep} To the left of $L_6$, place $2n^4$ embedding points on $L_4$, and 
 $2n^4$ embedding points on $L_3$ (See Figure \ref{msgg}).
These embedding points must be placed in such a way that 
$(a)$ each embedding point added to $L_3$ blocks an embedding point on $L_4$ from $p_2$,
$(b)$ each embedding point added on $L_3$ sees every embedding point of $G$ not on $L_3$, and
$(c)$ each embedding point added on $L_4$ sees every embedding point of $G$ not on $L_4$.
To achieve this, the embedding points on $L_4$ and $L_3$ 
are added by considering the intersections of $L_3$ and $L_4$ with lines containing the edges of $G$.
Such intersections are at most twice the number of edges in $G$. Each new embedding point can be placed on $L_4$ and its 
blocker on $L_3$ by avoiding these intersections.
  Add all edges between embedding points that are visible to each other.

\item Below $L_5$, place $25n^8$ embedding points on $L_2$, and
%
$25n^8$ embedding points on $L_1$ (See Figure \ref{msgg}).
These embedding points must be placed in such a way that $(a)$ each embedding point added to $L_1$ blocks
an embedding point on $L_2$ from $p_1$,
$(b)$ each embedding point added on $L_1$ sees every embedding point of $G$ not on $L_1$, and
$(c)$ each embedding point added on $L_2$ sees every embedding point of $G$ not on $L_2$.
To achieve this, the embedding points are added by following the method in step \ref{thrstep}.
%
 Add all edges between embedding points that are visible to each other.
%
\end{enumerate}
Henceforth $G$ is referred to as the \emph{modified slanted grid graph}, denoted as $MSGG$.
Observe that 
 for all pairs of embedding points $p_{i,j}$ and $p_{k,l}$ where $i \neq k$ and $j \neq l$,
$p_{i,j}$ and $p_{k,l}$ are mutually visible.

\subsection{Unique visibility embedding of MSGG} \label{unqemmsgg}
Here we prove that every $MSGG$ has
only one visibility embedding, unique up to the preservation of lines.
We start with a few properties.
 \begin{figure}[h] 
\begin{center} 
\centerline{\hbox{\psfig{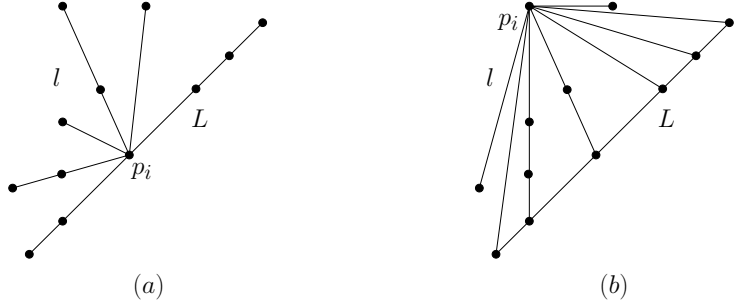}}}
\caption{(a) The point $p_i$ is on $L$ and hence sees at most $l+2$ embedding points. (a) The point
$p_i$ is not on $L$ and hence sees at least $k$ embedding points.}
\label{figsmallemma}
\end{center}
\end{figure}
\begin{lemma}\label{smallemma}
Let $H_1$ be a $PVG$ with visibility embedding $\xi$. Let $L$ be a line in  $\xi$
such that $(i)$ there are $k$ embedding points on $L$, and $(ii)$ $l$ embedding points not on $L$.
Let $v_i$ be a vertex of $H_1$ and $p_i$ its corresponding embedding point in $\xi$.
If $p_i$ lies on $L$, then $deg(v_i) \leq l+2$. Otherwise, 
$deg(v_i) \geq k$.
\end{lemma}
 \begin{proof} 
If $p_i$ lies on $L$ then $p_i$ sees at most 
two embedding points on $L$, and at most all $l$ embedding points that are not on $L$ (see Figure \ref{figsmallemma}(a)).
If $p_i$ does not lie on $L$, then all $k$ embedding points on $L$ lie on distinct lines passing through $p_i$.
So, $p_i$ sees at least $k$ embedding points (see Figure \ref{figsmallemma}(b)).
\end{proof}
\begin{lemma} \label{lp} 
Let $H_2$ be a $PVG$ with visibility embedding $\xi$. 
Let $L$ be a line in $\xi$ such that $(i)$ there are $k$ embedding points on $L$, and $(ii)$ $l$ embedding points not on $L$.
If $k \geq (l+3)^2$ then $L$ is preserved 
in every visibility embedding of $H_2$ (Figure \ref{fig1.5lp}(a)).
\end{lemma}
\begin{proof} 
By the hypotheses, $H_2$ has $k+l$ vertices. Let us assume on the contrary that $L$ is not preserved in some visibility 
embedding $\xi'$ of $H_2$. Let $\phi$ denote the bijection between $\xi$ and $\xi'$. Let $\phi(L)$ denote
the set of images of all embedding points lying on $L$, in $\xi'$. 
We have the following cases depending on the collinearity of the embedding points of $\phi(L)$.  
$\\ \\$
\textbf{Case 1:} All embedding points in $\phi(L)$ are collinear.
$\\$
\textbf{Case 2:} Not all embedding points in $\phi(L)$ are collinear.
$\\ \\$
Consider \emph{Case 1}. 
Let $L'$ be the line containing all embedding points of $\phi(L)$.
Consider the situation where $L'$ contains only the embedding points of $\phi (L)$.
 Let $p_{i-1}$, $p_i$ and $p_{i+1}$ be three consecutive embedding points on $L$
  whose corresponding vertices in $H_2$ are $v_{i-1}$, $v_i$ and $v_{i+1}$ respectively.
  Clearly,  $\phi(p_{i-1})$, $\phi(p_{i})$ and $\phi(p_{i+1})$ must be consecutive embedding points on $L'$,
  since $(v_{i-1},v_i)$ and $(v_i, v_{i+1})$ are edges of $H_2$. A similar argument holds for the first and last 
  embedding points of $L'$. Hence, $L$ is preserved.
  Consider the other situation where $L'$ contains an embedding point $p_i$ not in $\phi(L)$.
Let the corresponding vertex to $p_i$ in $H_2$ be $v_i$.
Since $p_i \notin L$ and $\phi(p_i)$ lies on $L'$, $k \leq deg(v_i) \leq l+2$ by Lemma \ref{smallemma},
contradicting the assumption that $k \geq (l+3)^2$.
$\\ \\$
Consider \emph{Case 2}. If not   all embedding points of $\phi(L)$ are collinear, 
then either  $(i)$ no $(l+3)$ embedding points of $\phi(L)$ 
are collinear, or  $(ii)$ some $(l+3)$ embedding points of $\phi(L)$ are collinear  .
%
%
Consider $(i)$. Let $p_i \in \phi(L)$. Since $\vert \phi(L) \vert \geq (l+3)^2$
by assumption and no $(l+3)$ embedding points of $\phi(L)$ are collinear, there are at least $(l+4)$ 
distinct lines passing through $p_i$. So, the degree of the corresponding vertex $v_i$ of $p_i$ in $H_2$ is at least $(l+4)$.
On the other hand, by Lemma \ref{smallemma}, $deg(v_i) \leq l+2$, a contradiction. 
$\\ \\$
Consider $(ii)$. 
Let $L'$ be a line containing $(l+3)$ embedding points of $\phi(L)$.
Let $p_i \in \phi(L)$ and $p_i \notin L'$ such that $p_i$ is closest to $L'$ among all embedding points of 
$\phi(L)$. Since $\phi ^ {-1} (p_i)$ sees at most two points on $L$, $p_i$ does not see at   least $l+1$ points.
Hence, $p_i$ requires $l+1$ blockers where no blocker is from $\phi(L)$ by the choice of $p_i$.
On the other hand, there are only $l$ points not in $\phi(L)$, a contradiction.
%
%
%
%
%
\end{proof}
 $\\$
Let $L_a$ and $L_b$ be two lines in a visibility embedding $\xi$ of a special type of PVG such that most of the embedding points of
$\xi$ are on $L_a$ and $L_b$ (Figure \ref{fig1.5lp}(b)). In the following lemma, 
we show that $L_a$ and $L_b$ are preserved in every visibility embedding 
of the PVG.
%
%
%
%
\begin{figure}
\begin{center} 
\centerline{\hbox{\psfig{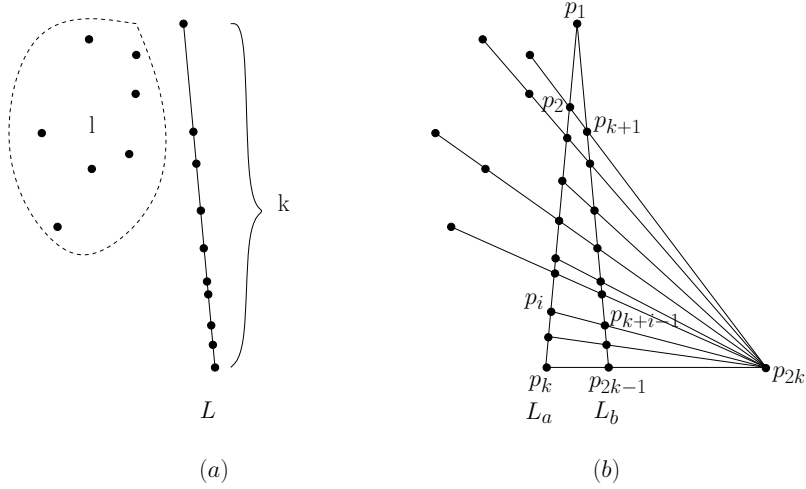}}}
\caption{(a) The line $L$ is preserved. (b) Both lines $L_a$ and $L_b$ are preserved.}
\label{fig1.5lp}
\end{center}
\end{figure}
\begin{lemma} \label {2lp}
Let $H_3$ be a PVG with visibility embedding $\xi$. Let
$L_a = \langle p_1,p_2,...,p_k\rangle$, and 
$L_b = \langle p_1,p_{k+1},...,p_{2k-1}\rangle$
be two lines in $\xi$ such that 
$k \geq (l+3)^2$, where $l$ denotes the number of embedding points in $\xi\setminus\{L_a\cup L_b\}$.
%
%
%
%
%
%
%
%
%
Let $p_{2k}$ be an embedding point satisfying the following properties.
\begin{enumerate}
\item   \label{a1.2} The embedding point $p_{2k}$ is adjacent to all embedding points in $L_b$, and is not adjacent
to any other embedding point of $\xi$.
\item \label{a1.1} For $1< i\leq k$, $p_{k+i-1}$ blocks $p_i$ from $p_{2k}$ .
\item  \label{a1.3} Every embedding point in $L_a\setminus \{p_1\}$ is adjacent to every embedding point in $L_b\setminus \{ p_1 \}$.
\item \label{a1.4}  No embedding point in $\xi \setminus (L_a \cup L_b \cup \{ p_{2k} \} )$ is adjacent to   any embedding point of $L_b$.
\end{enumerate}
Then $L_a$ and $L_b$ are preserved in every visibility embedding of $H_3$, and the embedding points in 
$\xi \setminus (L_a \cup L_b \cup \{ p_{2k} \} )$ lie outside the convex hull of $(L_a \cup L_b \cup \{ p_{2k} \})$.
\end{lemma}
\begin{proof} Let $\xi'$ be any other visibility embedding of $H_3$. 
Let $\phi$ denote the bijection between $\xi$ and $\xi'$. So, $\phi(L_a)$ and $\phi(L_b)$ are  
the images of $L_a$ and $L_b$ in $\xi'$, respectively.
We know that embedding points of $\phi(L_b)$ are adjacent to $\phi(p_{2k})$.
The order of embedding points along $\phi(L_b)$ must be the same as that of $L_b$, because otherwise,
the corresponding 
edges in the PVGs for $\xi$ and $\xi'$ are different, a contradiction.
Consider any three consecutive points 
$\phi(p_i)$, $\phi(p_{i+1})$ and $\phi(p_{i+2})$ of $\xi'$ on $\phi(L_b)$ (Figure \ref{fig2lp}(a)).
If $\phi(p_{i+1})$ is the blocker between $\phi(p_i)$ and $\phi(p_{i+2})$, then they
are collinear. 
Otherwise, consider the triangle $\phi(p_i) \phi(p_{i+2}) \phi(p_{2k})$.
Observe that $\phi(p_i) \phi(p_{i+2}) \phi(p_{2k})$ must be a triangle and not a line segment
with $\phi(p_{2k})$ in the middle,
for otherwise, the points of $\phi(L_b) \setminus \{ \phi(p_i)$, $\phi(p_{i+1})$ and $\phi(p_{i+2}\}$,
which are all visible from $\phi(p_{2k})$
must lie outside of segment $\phi(p_i)   \phi(p_{2k}) \phi(p_{i+2})$. Hence, they   
must be blocked from $\phi(p_i)$ or $\phi(p_{i+2})$. Clearly, there are not enough blockers to achieve this,
hence forcing $\phi(p_i) \phi(p_{i+2}) \phi(p_{2k})$ to be a triangle.
If $\phi(p_{i+1})$ lies inside the triangle, then some point from
$\xi' \setminus ( \phi(L_b) \cup \{ \phi(p_{2k}) \} )$ can block $\phi(p_i)$ and $\phi(p_{i+2})$.
Consider the other case where 
 $\phi(p_{i+1})$ lies outside the triangle.
  The blocker between 
 $\phi(p_i)$ and $\phi(p_{i+2})$ must be adjacent to $\phi(p_{2k})$, and only the points of 
 $\phi(L_b)$ are adjacent to $\phi(p_{2k})$.
 If points from  $\phi(L_b)$ act as blockers between $\phi(p_i)$ and $\phi(p_{i+2})$,
  the blockers must form a chain where consecutive points
see each other. If other points from $\phi(L_b)$ are used as blockers, then
this chain is broken at some point.
 So, there cannot be any blocker of $\phi(p_i)$ and $\phi(p_{i+2})$.
%
%
%
%
Hence, the points of $\phi(L_b)$ must either be collinear or form a 
reflex chain facing $\phi(p_{2k})$ (Figure \ref{fig2lp}(b)).
%
%
$\\ \\$
Before showing that $L_b$ is preserved, we show that $L_a$ is preserved.
Since the embedding points of $\phi(L_b)$ form a reflex chain or a straight line and they are 
the only embedding points adjacent to $\phi(p_{2k})$, no embedding point of $(\phi(L_b) \cup \{ \phi(p_{2k}) \})$
can be a blocker for any pair of the remaining embedding points of $\xi'$.
In addition, these embedding points are also not blockers between $\phi(p_1)$ and any other embedding point of $\xi'$.
So, applying Lemma \ref{lp} on $(\xi' \setminus ((\phi(L_b) \setminus \{ \phi(p_1) \}) \cup \{ \phi(p_{2k}) \} ))$,
we get that $L_a$ is preserved.
$\\ \\$
Since $L_a$ is preserved and $ \vert \phi(L_a) \vert = \vert \phi(L_b) \vert$, the 
embedding points of $\phi(L_a)$ cannot be blockers for pairs of 
embedding points of $\phi(L_b)$.
Observe that as no embedding point of $p_q \in (\xi' \setminus \{ \phi(L_a) \cup \phi(L_b) \cup \{ \phi(p_{2k}) \} \})$
is visible from $\phi(p_2k)$, if 
they lie inside the region bounded by $\phi(L_a)$ and $\phi(L_b)$, then they 
must lie on rays from $\phi(p_2k)$ passing through points of $\phi(L_a)$.
This blocks points of $\phi(L_a)$ and $\phi(L_b)$ from each other.
So, no embedding point $p_q \in (\xi' \setminus \{ \phi(L_a) \cup \phi(L_b) \cup \{ \phi(p_{2k}) \} \})$
can lie 
inside the region bounded by $\phi(L_a)$ and $\phi(L_b)$.
%
%
%
Therefore, $p_q$ cannot be a blocker for pairs of 
embedding points of $\phi(L_b)$.
Hence, the blockers for embedding points of $\phi(L_b)$ must come from $\phi(L_b)$ itself.
So, the points of $\phi(L_b)$ are collinear and $L_b$ is preserved.
\end{proof}
%
%
\begin{figure}
\begin{center} 
\centerline{\hbox{\psfig{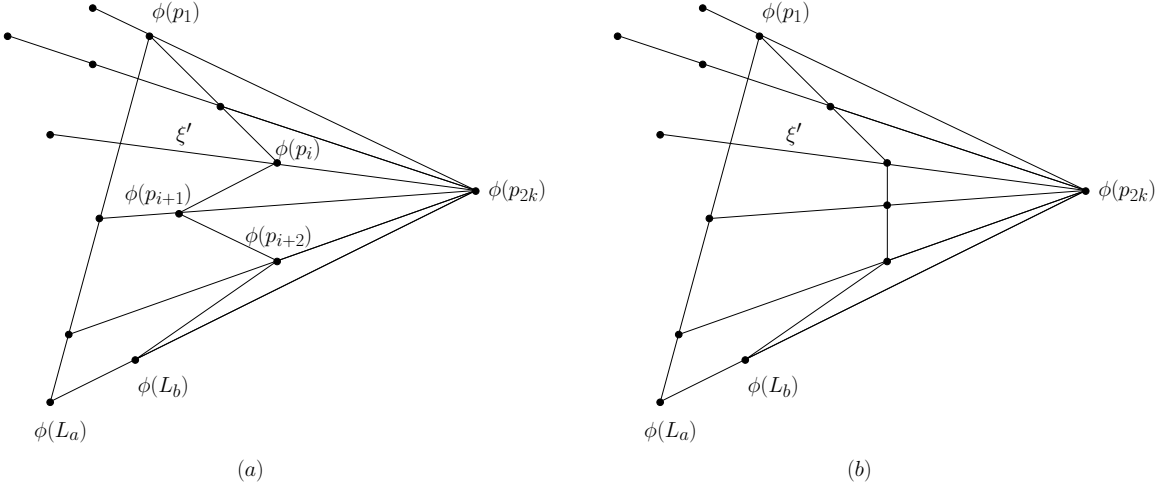}}}
\caption{
(a) In $\xi'$, $\phi(L_b)$ is not a reflex chain. 
(b) In $\xi'$, $\phi(L_b)$ a reflex chain facing $\phi(p_{2k})$. 
}
\label{fig2lp}
\end{center}
\end{figure}
 \begin{lemma} \label{4lp}
 Let $G$ be a modified slanted grid graph with visibility embedding $\xi$ (Figure \ref{msgg}).
 Let $L_1$ and $L_2$ be the rightmost and the second-rightmost lines in $\xi$, respectively.
 Let $L_3$ and $L_4$ be the topmost and the second-topmost lines in $\xi$, respectively. 
  The lines $L_1$, $L_2$, $L_3$ and $L_4$ 
  are preserved in every visibility embedding of $G$.  
 \end{lemma}
\begin{proof} 
Let $\xi'$ be any other visibility embedding of $G$. 
Let $\phi$ denote the bijection between $\xi$ and $\xi'$. So, $\phi(L_1)$, $\phi(L_2)$, $\phi(L_3)$ and $\phi(L_4)$ are  
the images of $L_1$, $L_2$, $L_3$ and $L_4$ in $\xi'$, respectively.
$\\ \\$
First we show that $L_1$ and $L_2$ are preserved.
$\phi(L_1) \setminus \{\phi(p_2)\}$ and $\phi(L_2)\setminus \{\phi(p_2)\}$ contain $25n^8 + n $
embedding points each by construction and $\xi' \setminus (\phi(L_1) \cup \phi(L_2))$
contains at most $4n^4 + n^2 - 2n - m_0^2 + 1 $ embedding points by construction, where  
$1 \leq \ m_0 \leq n -2$. Observe that $25n^8 + n + 1 \geq (4n^4 + n^2 - 2n - m_0^2 + 4)^2$ for large $n$. 
Since $k \geq (l+3)^2$, where $k = 25n^8 + n + 1$ and $ l = 4n^4 + n^2 - 2n - m_0^2 + 1$, both
$L_1$ and $L_2$ are preserved by Lemma \ref{2lp}.
%
$\\ \\$
Now we show that $L_3$ and $L_4$ are preserved.
Let us start by identifying the partition of blockers in $\xi'$ with respect to $\phi(L_1)$ and $\phi(L_2)$.
Without loss
of generality, let $\phi(p_2)$ be the topmost embedding point, 
$\phi(L_1)$ be to the right of $\phi(L_2)$ and $\phi(p_1)$ be to the right of $\phi(L_1)$ in $\xi'$ (see Figure \ref{msgg}). 
Since the adjacency relationships between  $\phi(L_1)$ and $\phi(L_2)$ cannot change, and 
$\phi(p_1)$ is adjacent only to the embedding points of 
$\phi(L_1)$, all embedding points of $\xi' \setminus (\phi(L_1) \cup \phi(L_2) \cup \{ \phi(p_1) \} )$ 
must be to the left side of $\phi(L_2)$.
Hence, embedding points of $ \phi(L_1) \cup \phi(L_2) \cup \{\phi(p_1)\}$ cannot be blockers 
for any pair of embedding points in $ \xi' \setminus (\phi(L_1) \cup \phi(L_2) \cup \{\phi(p_1)\} )$.
Embedding points of $\phi(L_3)$ must form a straight line or a reflex chain facing $\phi(p_2)$
as shown in Lemma \ref{2lp}.
Therefore the
embedding points of   $(\phi(L_1) \cup \phi(L_2) \cup \phi(L_3) \cup \{ \phi(p_1), \phi(p_2) \})$
 cannot be the blockers of the remaining 
embedding points of $\xi'$. 
Again, the set 
$\phi(L_4) \setminus (\phi(L_1) \cup \phi(L_2) \cup \{ \phi(p_1) \}) $
has $2n^4 + n - 2$ embedding points and
$\xi' \setminus (\phi(L_1) \cup \phi(L_2) \cup \phi(L_3) \cup \phi(L_4))$ has at most $(n-2)^2 - m_0^2$ embedding points.
Observe that $2n^4 + n - 2 \geq ((n-2)^2 - m_0^2 +3)^2$ for large $n$. 
Since $k \geq (l+3)^2$, where $k = 2n^4 + n - 2$ and $ l = ((n-2)^2 - m_0^2)$,
the embedding points of $\phi(L_4) \setminus (\phi(L_1) \cup \phi(L_2) \cup \{ \phi(p_1) \})$ 
are collinear in their original order by
Lemma \ref{lp} on  $\xi' \setminus (\phi(L_1) \cup \phi(L_2) \cup \phi(L_3) \cup \{ \phi(p_1), \phi(p_2) \})$.
%
$\\ \\$
We have already shown that $\phi(L_4) \setminus (\phi(L_1) \cup \phi(L_2) \cup \{ \phi(p_1) \})$ is a straight line.
If these embedding points are collinear with $\phi(p_1)$ and $\phi (L_4) \cap ( \phi (L_1) \cup \phi (L_2))$ (see Figure \ref{msgg}),
then $L_4$ is preserved. Otherwise, the embedding points of  $\phi(L_4) \setminus (\phi(L_1) \cup \phi(L_2) \cup \{ \phi(p_1) \})$
are collinear with $\phi(p_1)$ and $\phi(p_2)$, as $\phi(p_2)$ and $\phi(L_1) \cap \phi(L_4)$
are the only two embedding points of $\phi(L_1)$ that 
are not adjacent to all embedding points of $\phi(L_4) \setminus (\phi(L_1) \cup \phi(L_2) \cup \{ \phi(p_1) \})$.
%
%
Observe that since the embedding points of $\phi(L_4) \setminus (\phi(L_1) \cup \phi(L_2) \cup \{ \phi(p_1) \})$
form either a straight line or a reflex chain facing $\phi(p_2)$,
there cannot be any other embedding point on the line passing through $\phi(p_1)$ and $\phi(p_2)$.
So, the embedding points of $\phi(L_4) \setminus (\phi(L_1) \cup \phi(L_2) \cup \{ \phi(p_1) \})$
must lie on the line through $\phi(p_1)$ and $\phi(L_1) \cap \phi(L_4)$.
Furthermore, since the adjacency relationships between the embedding points of $\phi(L_4)$ cannot change, $L_4$ is preserved. 
Since all segments between embedding points of $\phi(L_4)$ and $\phi(p_2)$ require distinct
embedding points of $\phi(L_3)$, and $\vert \phi(L_3) \vert = \vert \phi(L_4) \vert$,
every embedding point of $\phi(L_3)$ must lie on the horizontal line passing through $\phi(p_1)$ and $\phi(L_1) \cap \phi(L_3)$.
Since they are all collinear and the adjacency relationships between the embedding points of $\phi(L_3)$ cannot change, $L_3$ is also preserved.
\end{proof} 
\begin{lemma} \label{uemsgg}
 Let $G$ be a modified slanted grid graph with visibility embedding $\xi$ (Figure \ref{msgg}).
$G$ has a unique visibility embedding, up to the preservation of lines.
\end{lemma}
\begin{proof} 
By Lemma \ref{4lp}, $L_1$, $L_2$, $L_3$ and $L_4$ are preserved. 
Let $\xi'$ be any other visibility embedding of $G$. 
Let $\phi$ denote the bijection between $\xi$ and $\xi'$. So, $\phi(L_1)$, $\phi(L_2)$, $\phi(L_3)$ and $\phi(L_4)$ are  
the images of $L_1$, $L_2$, $L_3$ and $L_4$ in $\xi'$, respectively.
$\\ \\$
Consider any horizontal line $L_i$ in $\xi$ passing through the embedding points $\{ p_1, p_{i_1}, p_{i_2}, \ldots, p_{i_j}\}$,
where $p_{i_1}$ and $p_{i_2}$ lie on $L_1$ and $L_2$ respectively. In $\xi'$,
all the embedding points of $\phi(L_1) \setminus \{ \phi(p_2) \cup \phi(p_{i_1}) \}$ 
are adjacent to all the embedding points of 
$\phi(L_i) \setminus \{\phi(p_{i_1}), \phi(p_{i_2})  \}$.
On the other hand, by the arguments 
of Lemma \ref{4lp}, the embedding points of $\phi(L_i) \setminus \{\phi(p_{i_1}), \phi(p_{i_2})  \}$ cannot
lie on the line passing through $\phi (p_1)$ and $\phi(p_2)$. 
Hence, the embedding points of $\phi(L_i) \setminus \{\phi(p_{i_1}), \phi(p_{i_2})  \}$
must lie on the horizontal line passing through $\phi (p_1)$ and $\phi(p_{i_1})$.
Since $L_1$ and $L_2$ are preserved, and $\vert L_1 \vert = \vert L_2 \vert$, 
$\phi(p_{i_2})$
must also lie on the horizontal line passing through $\phi (p_1)$ and $\phi(p_{i_1})$.
Since the adjacency relationships between the embedding points of $\phi(L_i)$ cannot change, 
the embedding points of $\phi(L_i)$ are collinear in the order of their pre-images in $L_i$.
This property is also true for all vertical lines and all other horizontal lines.
Hence, all horizontal and vertical lines of $\xi$ are preserved.
%
%
%
%
Consider a non-horizontal and non-vertical line passing through embedding points of $\xi$.
All such lines pass through exactly two embedding points of $\xi$ and it can be seen that these lines are also preserved.
Hence, $G$ has a unique visibility embedding, up to the preservation of lines.
\end{proof}
\section{A 3-SAT graph} \label{sec3sat}
In this section, we first construct a \emph{3-SAT graph} $G'$, corresponding to a 3-SAT formula $\theta$
 of $n$ variables $\{ x_1, x_2, \ldots, x_n \}$ and $m$ clauses $\{ C_1, C_2, \ldots, C_m \}$.
 Note that here $n$ and $m$ may denote quantities different from what they denote in   Sections \ref{secsgg} and \ref{secmsgg}.
 Then $G'$ is embedded into $G$ to construct a \emph{reduction graph} $G''$ such that 
 $G''$ is a PVG if and only if $\theta$ is satisfied.
 An embedding of $G''$ consists of regions called \emph{variable patterns} and \emph{clause patterns} respectively.
The number of clause patterns and variable patterns correspond to the number of clauses and variables respectively, in $\theta$.
\subsection{Construction of a 3-SAT graph} \label {consalg}

The construction of a 3-SAT graph $G'$ is described with respect to the unique visibility embedding $\xi$ of $G$.
Initially, $G$ is constructed from a slanted grid graph of $\alpha \times \alpha$ vertices, where $\alpha = 12(m+n)$,
by the process stated in Section \ref{consmsgg}. The large size of the 3-SAT graph will later be used to
enforce some collinearity conditions.
We know that the vertices of $G$ are placed as embedding points on the intersection points of horizontal and vertical lines of $\xi$.
Recall that there are intersection points in $\xi$ that do not contain any embedding point corresponding to the vertices of $G$.
We wish to use these free intersection points for embedding points corresponding to the vertices of $G'$.
Embedding points corresponding to vertices of $G'$ are placed on the free intersection points in such a way
that they correspond to the variables and clauses of $\theta$.
For every vertical line $l$ in $\xi$,
we refer to the embedding point on $l$ adjacent to $p_2$ as the \emph{topmost embedding point} of $l$,
and the next embedding point of $l$  
is called the \emph{second topmost embedding point} of $l$.
 The vertices of $G'$ are classified into the following six types.
 \begin{enumerate}
  \item \emph{\textbf{Occurrence vertices (o-vertices):}}
  Let $n_{i}$ and $\overline{n_{i}}$ be the number of clauses of $\theta$ in which $x_i$ and $\overline{x_i}$ occur, respectively.
  A group of vertices of size $(n_{i} + \overline{n_{i}} + 2)$ in $G'$ corresponding to $x_i$ and $\overline{x_i}$ in $\theta$
  are referred to as \emph{o-vertices} of $x_i$ in $G'$.
  The o-vertex corresponding to $x_i$ (or, $\overline{x_i}$) in $C_j$ is denoted as 
$o_{i,j}$ (respectively, $\overline{o}_{{i},j}$). Two more o-vertices of $x_i$
are denoted as $o_{i,0}$ and $\overline{o}_{{i},0}$ respectively.
  The embedding points corresponding to o-vertices are called
  \emph{o-points} (Figures \ref{varpat} and \ref{3SAT}). 
  For each $x_i$,  o-points are embedded on two distinct vertical lines of $\xi$ called the 
  \emph{left o-line} and \emph{right o-line} of $x_i$, respectively (Figures \ref{varpat} and \ref{3SAT}). 
  The left o-line and right o-line contain all the o-points corresponding to $x_i$ and $\overline{x_i}$, respectively.
 The o-points embedded on the left o-line (or, the right o-line) 
 are called the \emph{left o-points} (respectively, \emph{right o-points}) of $x_i$ and their corresponding vertices are called
 the \emph{left o-vertices} (respectively, \emph{right o-vertices}) of $x_i$.
   The o-points need to be blocked by l-points from their corresponding t-points (both described later).
 If too many l-points are utilized in blocking the o-points from t-points, then the l-points cannot be used to 
 block the visibility between c-points (also described later) from some other embedding points. This
 leads to unsatisfied visibility constraints.
   We denote the topmost embedding points of the left and right o-lines of $x_i$ by $x_{i,l}$ and $x_{i,r}$ respectively.
 
  \item  \emph{\textbf{Truth value vertices (t-vertices):}} For every variable $x_i$ there exists exactly one vertex 
  of $G'$ called the \emph{t-vertex} of $x_i$ (denoted as $t_i$), and its corresponding embedding point is called the \emph{t-point}
  of $x_i$ (Figures \ref{varpat} and \ref{3SAT}). For a given assignment of variables in $\theta$, $x_i$ can be $1$ or $0$.
%
%
  If $x_i=1$ (or, $0$), then the t-vertex of $x_i$ is embedded as the lowermost (respectively, uppermost) embedding point, 
 on the left (respectively, right) o-line of $x_i$.  
%
  If the t-point lies on its left o-line, then it needs to be blocked from its right o-points by some l-points,
and if the t-point lies on its right o-line, then it needs to be blocked from its left o-points by some l-points.

 \item \emph{\textbf{Clause vertices (c-vertices):}} 
 For every clause $C_j$, there exists exactly one vertex 
  of $G'$ called the \emph{c-vertex} of $C_j$ (denoted as $c_j$), and its corresponding embedding point is called the \emph{c-point}
  of $C_j$ (Figures \ref{clpat} and \ref{3SAT}).
  The rightmost vertical line of the clause pattern of a clause $C_j$ is called the \emph{c-line} of $C_j$.
 The c-point of $C_j$ is embedded as the lowermost embedding point of the c-line of $C_j$.
    We denote the topmost and second topmost embedding points of the c-line of $c_j$ as $c_{j,1}$ and $c_{j,2}$ respectively.
   The c-point $c_j$ needs to be blocked from  $c_{j,2}$ by an l-point. This is 
 possible only when all the l-points corresponding to $C_j$ are not blocking their corresponding t-points from their o-points.
  \item \emph{\textbf{Literal vertices (l-vertices):}} 
  These vertices also correspond to the occurrence of a variable and its complement in the clauses of $\theta$, and their
  corresponding embedding points are called \emph{l-points} (Figures \ref{varpat}, \ref{clpat} and \ref{3SAT}).
  Visibility of a t-point of a variable needs to be blocked 
  from the o-points of one of its o-lines. Visibility of the c-point
  of a clause also needs to be blocked from the second topmost embedding point of the vertical line in which it is embedded. The 
  l-points are used 
  as blockers in these cases. 
  An l-point corresponding to $x_i$ occurring in $C_j$ can be used to block either the t-point of $x_i$ from a left o-point
  of $x_i$, or the c-point of $C_j$ from  $c_{j,2}$ (which is necessary for
  satisfying a clause). 
  The  unique l-point (as well as l-vertex) corresponding to the embedding points used for blocking the visibility of 
  $t_i$ from $o_{i,j}$ (or, $\overline{o}_{{i},j}$), is denoted as 
$l_{i,j}$ (respectively, $\overline{l}_{{i},j}$). 
In a given assignment of $\theta$, if variable $x_i$ is assigned $1$, then the corresponding $l_{i,j}$ may be used
to block the corresponding c-point from the embedding point immediately above it. Otherwise, if $x_i$ is assigned $0$,
then $l_{i,j}$ has to block   $t_i$ from $o_{i,j}$.
The blocking is explained in details later.
 If $l_{i,j}$ blocks   the t-point of $t_i$ from the o-point of $o_{i,j}$, then $l_{i,j}$
must be embedded in a vertical line 
in the variable pattern of $x_i$, called the \emph{associated-line} of $l_{i,j}$. We denote the topmost embedding point 
this associated-line by $l_{i,j,1}$.
Similarly, we denote the topmost embedding point 
the associated-line of $\overline l_{i,j}$ by $\overline l_{i,j,1}$.
  \item \emph{\textbf{Dummy vertices (d-vertices):}} 
  For each variable $x_i$, there is exactly one vertex in $G'$ called the \emph{d-vertex}
  of $x_i$ (denoted as $d_i$), and its corresponding embedding point in $\xi$ is called the \emph{d-point} of $x_i$
  (Figures \ref{varpat} and \ref{3SAT}).
  The d-points are sometimes
  required to block the visibility of the right o-points   from the second topmost embedding point of their vertical line.
  The rightmost vertical line of the variable pattern of $x_i$ is called the \emph{d-line} of $x_i$.
  If $x_i$ is assigned $1$, then the d-point is embedded on the right o-line of $x_i$. Otherwise it is embedded on the
  d-line of $x_i$.   We denote the topmost embedding point of the d-line of $d_i$ as $d_{i,1}$.
  \item \emph{\textbf{Blocking vertices (b-vertices):}}
  These vertices of $G'$ correspond to
  embedding points called \emph{b-points} (Figures \ref{varpat}, \ref{clpat} and \ref{3SAT}).
  The b-points are required to block the visibility between $(a)$  t-points and the topmost embedding points of o-lines,
   $(b)$ d-points and the topmost points of their 
  respective d-lines and right o-lines,
  t-points and the topmost embedding points of o-lines,
  $(c)$ l-points and the topmost embedding points 
  of   their respective associated-lines and c-lines, and   $(d)$ 
   d-points and o-points. 
  The vertical lines on which b-points are embedded are called \emph{b-lines}.
  $\\ \\$
  There are   the following four types of b-vertices.
  \begin{enumerate} [(i)]
 \item The b-vertices corresponding to b-points that are used to block the visibility of embedding points of $t_i$ and $d_i$ 
  from 
   $x_{i,l}$, $x_{i,r}$ and $d_{i,1}$,
  are denoted as $b^1_{i}$,  $b^2_{i}$,  $b^3_{i}$ and  $b^4_{i}$.
  If the t-point of $t_i$ is embedded on its left o-line (or, right o-line), then the b-point of $b^2_{i}$ (respectively, $b^1_{i}$)
  blocks the t-point from $x_{i,r}$ (respectively, $x_{i,l}$).
  If the d-point of $d_i$ is embedded on its right o-line (or, d-line), then then $b^4_{i}$ (respectively, $b^3_{i}$)
  blocks the d-point from $d_{i,1}$ (respectively, $x_{i,r}$).
  The b-point corresponding to $b^1_{i}$ is embedded on the vertical line immediately to the left of the left o-line of $x_i$.
  The b-points corresponding to  $b^2_{i}$ and $b^3_{i}$ are embedded on the vertical lines immediately to the left and right of 
  the right o-line of $x_i$, respectively.
  The b-point correponding to $b^4_{i}$ is embedded on the vertical line immediately left to the d-line of $x_i$.

 \item The b-vertex corresponding to the b-point that is used to block the visibility of the l-point of 
   $l_{i,j}$ (or, $\overline{l}_{i,j}$)
  from 
   $l_{i,j,1}$ (respectively, $\overline{l}_{i,j,1}$)
    when the l-point of $l_{i,j}$ (respectively, $\overline{l}_{i,j}$)
  is embedded on its corresponding c-line and blocks the c-point from the second topmost point of the c-line,
  is denoted as $b_{i,j}$ (respectively, $\overline{b}_{{i},j}$).
  For every $l_{i,j}$ (or, $\overline{l}_{{i},j}$), the b-point of $b_{i,j}$ (respectively, $\overline{b}_{{i},j}$) 
  is embedded on the vertical line immediately to the right of the
  associated-line of $l_{i,j}$ (respectively, $\overline{l}_{{i},j}$).
  
 \item The b-vertex corresponding to the b-point that is used to block the d-point of $d_i$ from the o-point of
  $\overline{o}_{i,j}$, when the d-point is embedded on its d-line is denoted as $\overline{b}^d_{{i},j}$.  
  These b-points are embedded on vertical lines to the right of the right o-line of $x_i$. 
\item    The b-vertices corresponding to the c-line of $C_j$ are denoted as $b^c_{j,1}$ and $b^c_{j,2}$.
  The b-points corresponding to $b^c_{j,1}$ and $b^c_{j,2}$ are used to block $c_{j,1}$ from the l-points of $C_j$
  that are embedded on their associated lines of their respective variable patterns, 
  The   b-points corresponding to $b^c_{j,1}$ 
  and $b^c_{j,2}$ are 
    embedded on the two vertical lines immediately to the left of the c-line of $C_j$.
  \end{enumerate}
  \end{enumerate}
\begin{figure}
\begin{center} 
\centerline{\hbox{\psfig{figure=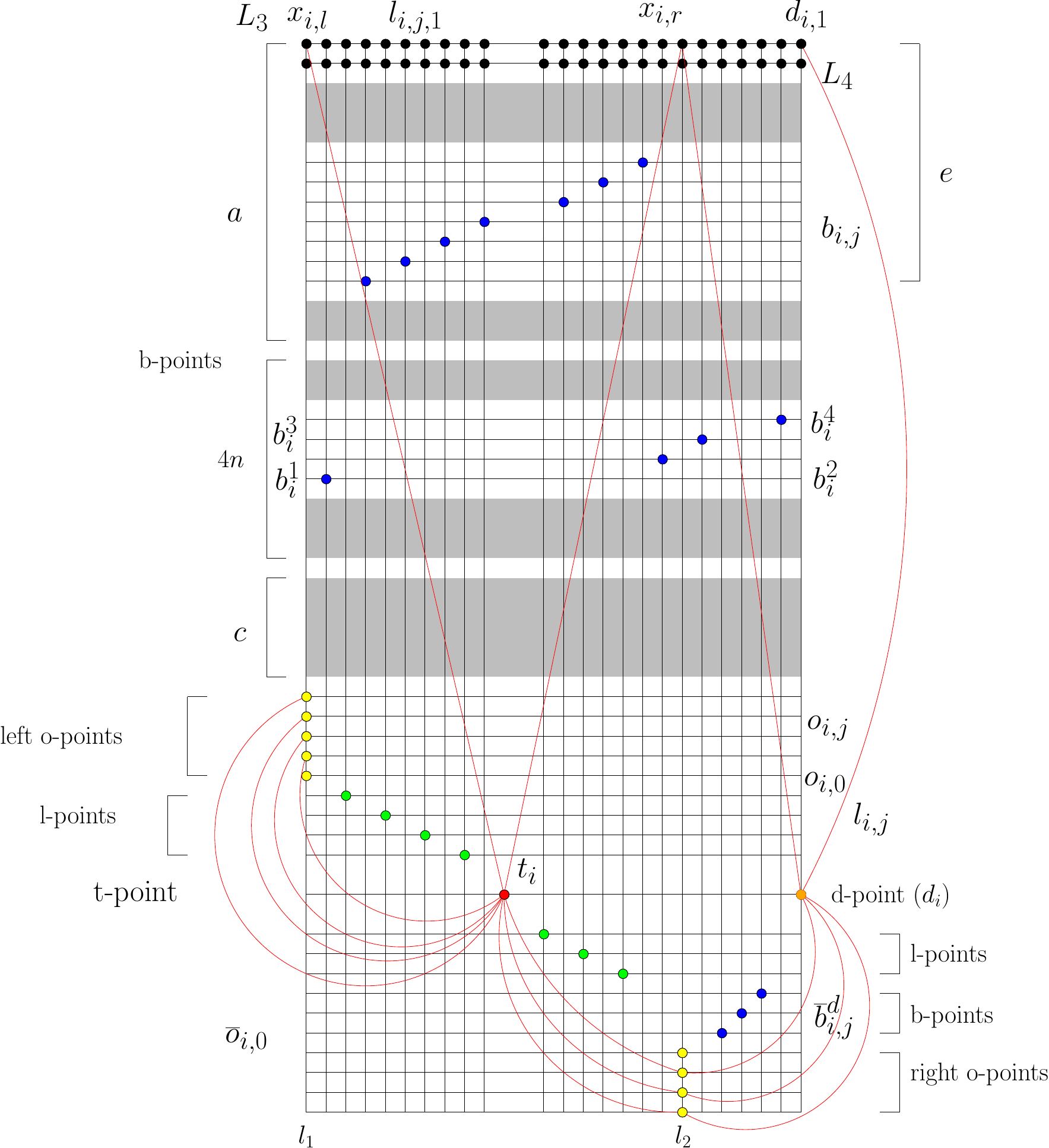,width=0.790\hsize}}}
\caption{
The variable pattern for $x_i$. 
The t-point, o-points, l-points, d-point and b-points are shown in the figure,
along with the names of their corresponding vertices.
The gray areas represent multiple horizontal lines.
Some non-edges   are drawn as curves. The lines $l_1$ and $l_2$ are the left and right o-lines 
of $x_i$ respectively.
}
\label{varpat}
\end{center}
\end{figure}
\begin{figure}
\begin{center} 
\centerline{\hbox{\psfig{figure=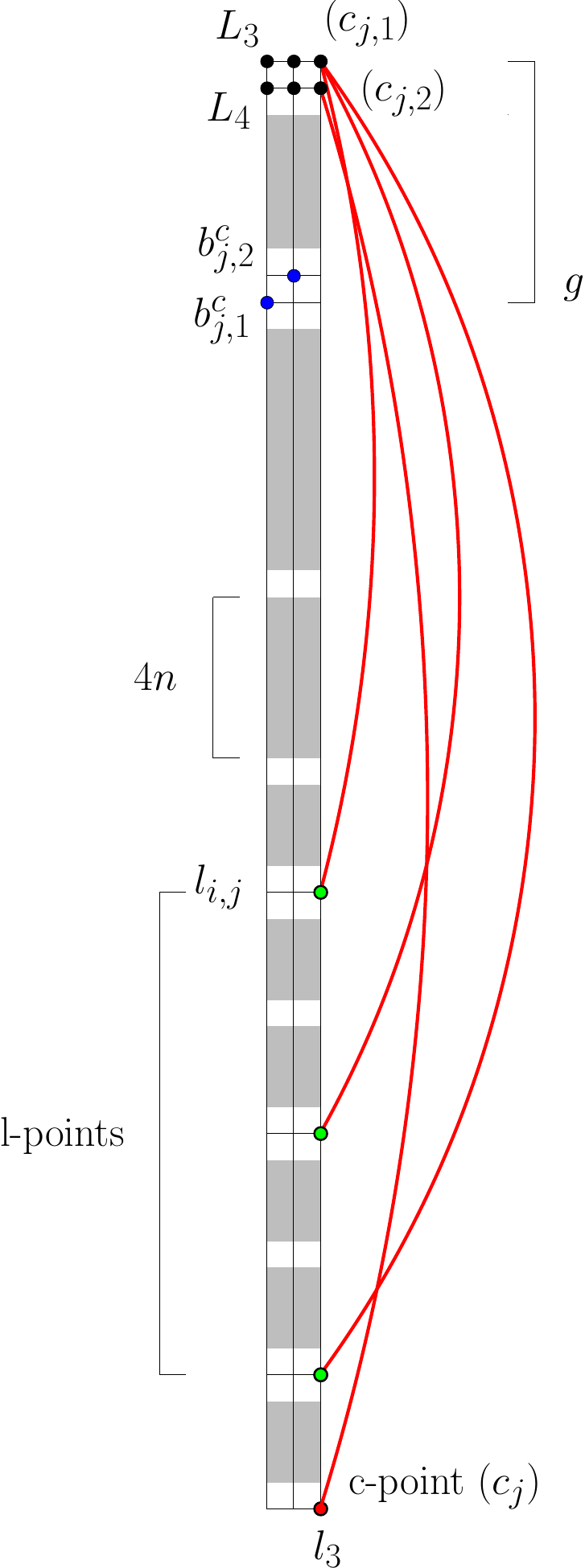,width=0.34\hsize}}}
\caption{
The clause pattern for $C_j$.
The c-point, l-points and b-point   are shown in the figure, along with the names of their corresponding vertices.
The gray areas represent multiple horizontal lines.
Non-edges are drawn as curves.
The line $l_3$ is the c-line of $C_j$.
%
}
\label{clpat}
\end{center}
\end{figure}
Based on the above classifications, we present the construction of $G'(V',E')$. The sets of vertices   
associated with $x_i$ and $C_j$ are denoted as $V^x_i$ and $V^c_j$ respectively.
Each $V^x_i$ contains the vertices $\{ t_i, d_i, o_{i,0}, \overline{o}_{{i},0}, b^1_{i},  b^2_{i},  b^3_{i},  b^4_{i} \}$.
For every $C_j$ containing $x_i$ (or, $\overline{x_i}$), $V^x_i$ contains the vertices
$\{ o_{i,j},l_{i,j},b_{i,j}  \}$
(respectively, $\{ \overline{o}_{{i},j}, \overline{l}_{{i},j}, \overline{b}_{{i},j}, \overline{b}^d_{{i},j}\}$).
Each $V^c_j$ contains $c_j$, $b^c_{j,1}$, $b^c_{j,2}$, and 
$
l_{{i},j} 
$ 
(or, 
$
\overline{l}_{{i},j}
$) 
corresponding to each $x_i$ (respectively, $\overline{x_i}$)
in $C_j$. Hence,
\begin{eqnarray*}
 V^x_i &=& \{ t_i, d_i, o_{i,0}, \overline{o}_{{i},0} \} \cup ( \bigcup_{x_i \in C_j} \{ o_{i,j},l_{i,j},b_{i,j}  \} )
 \cup ( \bigcup_{\overline{x_i} \in C_j} \{ \overline{o}_{{i},j}, \overline{l}_{{i},j}, \overline{b}_{{i},j}, \overline{b}^d_{{i},j}\} )
 \cup \{ b^1_{i},  b^2_{i},  b^3_{i},  b^4_{i}\}\\ 
 V^c_j &=&   \{ c_j \} 
  \cup ( \bigcup_{x_i \in C_j} 
  l_{{i},j} 
  ) \cup
 (\bigcup_{\overline{x_i} \in C_j} 
 \overline{l}_{{i},j} 
  ) \cup \{ b^c_{j,1}, b^c_{j,2} \} \\
 V' &=& (\bigcup_{x_i \in \theta} V^x_i) \cup (\bigcup_{C_j \in \theta} V^c_j) \\
\end{eqnarray*}
%
\noindent For every $x_i, x_k \in \theta$, each vertex of $V^x_i$ is adjacent to every vertex of $V^x_k$.
Similarly, for every $C_j, C_k \in \theta$, each vertex of $V^c_k$ is adjacent to each vertex of $V^c_k$.
The set of edges among vertices
of $V^x_i$ (or, $V^c_j$) is denoted as $E^x_i$ (respectively, $E^c_j$). 
All vertices of $V^x_i$ are adjacent to each other except the o-vertices, and all left o-vertices
are adjacent to all right o-vertices. The left o-vertices (or, right o-vertices) along with $t_i$ induce a path in $G'$. 
The right o-vertices and $d_i$ induce a path in $G'$ as well.
All vertices of $E^c_j$ are adjacent to each other. 
For every $x_i, C_j \in \theta$, each vertex of $V^x_i$ is adjacent to each vertex of $V^c_j$.
Hence,  
\begin{eqnarray*}
 E^x_i &=&  \Big\{ (v_a,v_b) \mid v_a \neq v_b \ \textrm{and}  \ v_a,v_b  \in V^x_i \setminus \Big( \{ o_{i,0}, \overline{o}_{{i},0} \}
 \cup  ( \bigcup_{x_i \in C_j} \{ o_{i,j} \}) \cup ( \bigcup_{\overline{x_i} \in C_j} \{ \overline{o}_{{i},j} \}) \Big) \Big\} \\
  & & \cup \Big\{ (v_a,v_b) \mid v_ a \in \{ o_{{i},0} \} \cup  ( \bigcup_{x_i \in C_j} \{ o_{i,j} \})
  \ \textrm{and} \ v_b \in \{ \overline{o}_{{i},0}  \} \cup ( \bigcup_{\overline{x_i} \in C_j} \{ \overline{o}_{{i},j} \}) \Big\}\\
 & &\cup   \Big\{ (v_a,v_b) \mid v_ a \in \{ o_{i,0}, \overline{o}_{{i},0} \}
 \cup  ( \bigcup_{x_i \in C_j} \{ o_{i,j} \}) \cup ( \bigcup_{\overline{x_i} \in C_j} \{ \overline{o}_{{i},j} \})
 \ \textrm{and}  \\
 & & v_b  \in V^x_i \setminus ( \{ t_i, d_i, o_{i,0}, \overline{o}_{{i},0} \}
 \cup  \Big( \bigcup_{x_i \in C_j} \{ o_{i,j} \}) \cup ( \bigcup_{\overline{x_i} \in C_j} \{ \overline{o}_{{i},j} \})  \Big) \Big\} \\
 & &\cup \Big\{ (t_i,o_{i,0}), (t_i,\overline{o}_{{i},0}), (d_i,\overline{o}_{{i},0}) \Big\}
 \cup \Big\{ (d_i,v_b) \mid v_b \in ( \bigcup_{x_i \in C_j} \{ o_{i,j} \}) \Big\} \\ 
 & & \cup \Big\{ (o_{i,j},o_{i,k}) \mid j < k \ \textrm{and} \ \nexists o_{i,l} \in V^x_i : i<l<k \Big\} \\
  & & \cup \Big\{ (\overline{o}_{i,j},\overline{o}_{i,k}) \mid j < k \ \textrm{and} \ \nexists \overline{o}_{i,l} \in V^x_i : i<l<k \Big\} \\
 E^c_j &=& \Big\{ (v_a, v_b) \mid v_a \neq v_b \ \textrm{and} \ v_a,v_b \in V^c_j \Big\} \\
 E' &=& \Big\{ \bigcup_{x_i\in \theta}  E^x_i \Big\} 
 \cup \Big\{ \bigcup_{C_j\in \theta} E^c_j \Big\} 
 \cup \Big(\bigcup_{x_i\in \theta}\{ (v_a,v_b) \mid v_a \in V^x_i \ \textrm{and} \ v_b \in V' \setminus V^x_i   \}\Big) \\
 & & \cup \Big(\bigcup_{C_j\in \theta} \{ (v_a,v_b) \mid v_a \in V^c_j \ \textrm{and} \ v_b \in V' \setminus V^c_j \} \Big)
\end{eqnarray*}
We have the following lemma on the size of $G'$.
\begin{lemma}\label{num}
There are   $(8n + 12m + \sum\limits_{x_i \in \theta}\overline{n_i})$ vertices in $G'$.
\end{lemma}
\begin{proof}
We know from the construction of $G'$ that  
the number of t-vertices is $n$,
the number of d-vertices is $n$,
the number of o-vertices is $(3m+2n)$,
the number of l-vertices is $3m$,
the number of c-vertices is $m$,
and  
the number of b-vertices is   $(4n + 5m + \sum\limits_{x_i \in \theta}\overline{n_i})$
in $G$.
So, $G'$ has a total of  
$(8n + 12m + \sum\limits_{x_i \in \theta}\overline{n_i})$ vertices.
%
\end{proof}
\subsection{Construction of a reduction graph}\label{subsecconsred}

Here, we construct the reduction graph $G''(V'',E'')$ 
such that 
 $G''$ is a PVG if and only if $\theta$ is satisfied.
From Lemma \ref{num},
we know that the number of vertices in $G'$ is   $(8n + 12m + \sum\limits_{x_i \in \theta}\overline{n_i})$.
To get $G''$ with certain restrictions on its possible visibility embeddings, we need to join
$G'$ to a modified slanted grid graph $G$ with edges such that  $\vert G \vert \sim \vert G' \vert ^{8}$.
%
  Let $\alpha = 12(m+n)$.
Since   $\alpha > (8n + 12m + \sum\limits_{x_i \in \theta}\overline{n_i})$, 
   but we do not actually require so many vertical lines to embed the 3-SAT graph,
the modified slanted grid graph $G(V,E)$ is constructed
%
starting from a $\alpha \times \alpha$ slanted grid graph, and   $m_0=3m$, as stated
in Section \ref{consmsgg}.
$\\ \\$
The vertices of $G''$ are the vertices of $G$ and $G'$. Hence,
 $V'' = V \cup V'$. 
Consider the unique visibility embedding of $G$, with $p_1$ and $p_2$ as the rightmost and topmost embedding points, respectively.
The $i^{th}$ horizontal line from the top 
and 
the $i^{th}$ vertical line from the left 
are denoted by $l^h_i$ and $l^v_i$ respectively.
The vertex of $V$ that corresponds to the embedding point at the intersection of the $i^{th}$ vertical and $j^{th}$ horizontal lines, 
is denoted by $v(l^v_i, l^h_j)$. Now we assign similar coordinates to vertices of $V'$. Note that we may assign 
 a set of possible coordinates to the same vertex, in order to facilitate the analysis of embeddings of $G''$.
$\\ \\$
%
 For each variable $x_i \in \theta$, coordinates are assigned to vertices of $V^x_i$ as follows (Figure \ref{varpat}).
 \begin{enumerate}[(i)]
  \item 
  Corresponding to each variable $x_i$, distinct horizontal lines are occupied by $n_i + \overline{n_i} + 2$ o-points,
  $n_i + \overline{n_i} $ l-points and $n_i + 2\overline{n_i} + 4$ b-points. The points $t_i$ and $d_i$ 
  occupy the same horizontal line. Hence, for all variables, a total of $9m +  \sum\limits_{x_i \in \theta}\overline{n_i}+ 7n$
  horizontal lines are occupied by embedding points. There are six   points for each clause. However, the three l-points
  for each clause are already counted for their corresponding variables, and the c-points of all the clauses occupy
  the same horizontal line. So, for all the clauses, a total of   $2m + 1$ horizontal lines are occupied excluding 
  those occupied by the l-points. Hence, a total of   $9m +  \sum\limits_{x_i \in \theta}\overline{n_i}+ 7n + 2m + 1$
  $= 11m +  \sum\limits_{x_i \in \theta}\overline{n_i} + 7n +1$ horizontal lines are occupied.
  Let $a$ be a constant which marks the coordinates of the lowest b-point among all $b_{i,j}$ for all $i$ and $j$.
  There are $3m$ number of b-points of the form $b_{i,j}$, and the b-points for the clauses,   $2m$ in number, 
  occupy higher horizontal coordinates. 
  So,   $a = \alpha - (11m +  \sum\limits_{x_i \in \theta}\overline{n_i} + 7n +1) + 3m +2m
  = \alpha - 6m - \sum\limits_{x_i \in \theta}\overline{n_i} -7n -1$. 
  $\\ \\$
  Let $b$ denote the number of vertical lines occupied by the embedding points for the first $i-1$ variables.
  The two large horizontal lines at the top of any visibility embedding of $G$ have $2(\alpha)^4 + \alpha$
  embedding points each, and they occupy the same number of vertical lines. So,
  before the main grid structure where points corresponding to $G'$ begins, $2(\alpha)^4$
  vertical lines are occupied. Also, the embedding points of each variable $x_i$ occupy 
     $(2n_{k} + 3\overline{n_{{k}}} +7)$ vertical lines. So,
  $b = 2(\alpha)^4 + \sum _{k=1} ^{i-1}(2n_{k} + 3\overline{n_{{k}}} +7)$.
 $\\ \\$
Let, $c$ denote the number of horizontal lines occupied by embedding points corresponding to the first $i-1$ variables.
Hence, by the previous calculations,
     $c=  \sum _{k=1} ^{i-1} (2n_{k} + 3\overline{n_{{k}}} + 3)$.
     For the time being, intuitively consider a variable pattern for the $i^{th}$ variable to be the region bounded by
     $l^v_{b+1}$,
 $l^v_{b+ 2n_{i} + 3\overline{n_{i}} +7}$,
  $L_3$
  and 
$l^h_{a+ 4n+c + 2n_{i} +2\overline{n_{i}} + 3 + \overline{n_i}}$, though
we describe the details of the corresponding embedding only in the next section.
Note that the y-coordinates of the lowermost horizontal lines of successive variable patterns 
give rise to a staircase-like structure as seen in Figure \ref{3SAT}. 
  \item 
  The t-point may lie only on one of the two o-lines, and its x-coordinates correspond to those of the two o-lines.
  Horizontally, it lies below the $n_i +1$ left o-points.
  So, assign coordinates $(l^v_{b+1},l^h_{a+ 4n+c + 2n_{i} +2 })$
and  $(l^v_{b+2n_{i} +2\overline{n_{i}}+4},l^h_{a+ 4n+c + 2n_{i} +2 })$ to $t_i$. 
  \item 
  The points $o_{i,0}$ and $\overline o_{i,0}$ are the bottommost and topmost embedding points of the left and right o-lines, respectively.
  So, assign coordinates $(l^v_{b+1},l^h_{a+ 4n+c + n_{i} + 1 })$
  and $(l^v_{b+2n_{i}+4},l^h_{a+ 4n+c + 2n_{i} +2\overline{n_{i}} + 3})$
  to $o_{i,0}$ and $\overline o_{i,0}$ respectively. 
  \item 
  The left o-line occupies the leftmost vertical line of a variable pattern. Its 
  o-points lie on the consecutive horizontal lines, beginning from immediately below the lowest horizontal line 
  of the $(i-1)^{th}$ variable pattern.
  Recall that the left o-vertices of $x_i$ induce a path along with $t_i$ in $G'$. Let 
  $S = (o_{i,l}, \ldots, o_{i,0}, t_i)$ be the sequence of vertices in the path. Note that 
  $S$ has $n_i +2$ elements. 
  So, for each $o_{i,j} \in S$, $j \neq 0$, assign the coordinates $(l^v_{b+1},l^h_{a+ 4n+c + k})$
  to $o_{i,j}$, where $o_{i,j}$ is the $k^{th}$ element of $S$.
   \item
   The right o-line occupies a vertical line after the left o-line, 
   all $n_i + \overline{n}_i$ l-points of the variable pattern, one b-point for each of the l-points, plus 
   two more b-points lie on one vertical line each.
    Its o-points lie on the consecutive horizontal lines, beginning from immediately below the horizontal line for for the lowest
    b-point of the variable pattern, described later.
   Similar to the left o-vertices, the right o-vertices of $x_i$ induce a path along with $t_i$ in $G'$. Let 
  $\overline{S} = (t_i, \overline o_{i,0}, \ldots, \overline o_{i,l})$ be the sequence of vertices in the path. Note that 
  $S$ has $\overline{n_i} +2$ elements. 
  So, for each $\overline o_{i,j} \in \overline S$, $j \neq 0$, assign the coordinates
  $(l^v_{b+2n_{i}+ 2\overline{n_{i}} +4}, l^h_{a+ 4n+c + 2n_{i} +2\overline{n_{i}} + 3 + k})$
  to $\overline o_{i,j}$, where $\overline o_{i,j}$ is the $(k+2)^{th}$ element of $\overline S$.
\item 
The l-points corresponding to the left o-points, lie on horizontal lines 
starting from immediately below the left o-points. If they lie inside the variable pattern at all, then they lie on vertical lines
starting from the third leftmost vertical line of the variable pattern, leaving a vertical line
in between each consecutive l-points, for a corresponding b-point.
To each l-vertex $l_{i,j}$, assign coordinates $(l^v_{b+1+2k},l^h_{a+ 4n+c + n_{i} + 1 +k })$,
where  $o_{i,j}$ is the $k^{th}$ element of $S$.
The line $l^v_{b+1+2k}$ is called an \emph{associated-line} of $l_{i,j}$.
\item 
The l-points corresponding to the right   o-points, lie on horizontal lines 
starting from immediately below the t-point. If they lie inside the variable pattern at all, they lie on vertical lines
starting from two vertical lines to the right of the vertical line containing the rightmost left l-point of the variable 
pattern, leaving a vertical line
in between each consecutive l-points, for a corresponding b-point.
So, to each l-vertex $\overline l_{i,j}$ assign coordinates $(l^v_{b+2n_{i}+ 1 + 2k},l^h_{a+ 4n+c + 2n_{i} + 2 + k })$,
where  $\overline o_{i,j}$ is the $(k+2)^{th}$ element of $\overline S$. 
The line $l^v_{b+1+2k}$ is called an associated-line of $\overline l_{i,j}$.
  \item 
  The d-point $d_i$ lies in the same horizontal line as that of $t_i$, and either on the 
  right o-line, or on the rightmost vertical line of the variable pattern.
  So, assign coordinates  $(l^v_{b+2n_{i}+2\overline{n_{i}}+4},l^h_{a+ 4n+c + 2n_{i} +2 })$
  and $(l^v_{b+ 2n_{i} + 3\overline{n_{i}} +7},l^h_{a+ 4n+c + 2n_{i} +2 })$ to $d_i$.
 \item 

 Assign coordinates $(l^v_{b+2}, l^h_{a+4(n-i)+4})$,
 $(l^v_{2(n_{i}+ \overline{n_{i}})+3}, l^h_{a+4(n-i)+3} )$,
 $(l^v_{2(n_{i}+ \overline{n_{i}})+5} , l^h_{a+4(n-i)+2})$ and \newline
 $(l^v_{2n_{i}+ 3\overline{n_{i}}+6}, l^h_{a+4(n-i)+1})$ to the b-vertices $b^1_{i}$, $b^2_{i}$, $ b^3_{i}$ and $b^4_{i}$ respectively.
 \item
 Let   $e = \alpha + n -  (8n + 12m + \sum\limits_{x_i \in \theta}\overline{n_i}) +
 \sum _{p=i} ^{n} (n_{p} + n_{\overline{p}}) + 5m - 1$.
 The b-points of the forms $ b_{i,j}$ and $\overline b_{i,j}$ lie on consecutive horizontal lines
 starting from $l^h_e$ in a bottom to top manner. Between each such b-point, there is a vertical line
 for accommodating an l-point.
 Assign coordinates $(l^v_{b+2 + 2k},l^h_{e  +1 -k})$ to $b_{i,j}$,
 where  $o_{i,j}$ is the $k^{th}$ element of $S$.
 Similarly, assign coordinates $(l^v_{b+2n_i + 2 + 2k},l^h_{e + \overline{n_{i}} +1 -k})$ to $\overline b_{i,j}$,
 where  $\overline o_{i,j}$ is the $(k+2)^{th}$ element of $\overline S$.
 \item 
  The lowest group of b-points of the variable pattern lie in between the right o-line and the 
 rightmost vertical line of the variable pattern.
 Their function is to block $d_i$ from the right o-points when $d_i$ lies on the rightmost vertical line of the variable pattern.
 So, assign coordinates 
 $\{ (l^v_{2(n_{i}+ \overline{n_{i}})+5+k},   l^h_{a+ 4n+c + 2n_{i} +2\overline{n_{i}} + 3 -k} )$
 to each $\overline{b}^d_{i,j}$, where
  $\overline o_{i,j}$ is the $(k+2)^{th}$ element of $\overline S$.
 \end{enumerate}
 From the assignment of coordinates to the above vertices,
%
%
for   the time being, intuitively consider the $j^{th}$ clause pattern to be the region bounded by the lines
$l^v_f$,
 $l^v_{f+3}$,
$L_3$
and 
$l^h_{\alpha}$
though
we describe the details of the corresponding embedding only in the next section.
 For each clause $C_j \in \theta$ , coordinates are assigned to the vertices of $V^c_j$ as follows (Figure \ref{clpat}).
   \begin{enumerate} [(i)]
   \item  
   Let $f$ denote the x-coordinate of the rightmost vertical line of the $(j-1)^{th}$ clause pattern. 
   The rightmost vertical line of the $n^{th}$ variable pattern has x-coordinate
   $2(\alpha)^4 (\sum _{k=1} ^{n} (2n_{k} + 3n_{\overline{k}} +7))$.
   The first $j-1$ clause patterns occupy   $3(j-1)$ vertical lines. 
   So,   $f =2(\alpha)^4 + (\sum _{k=1} ^{n} (2n_{k} + 3n_{\overline{k}} +7)) + 3(j-1)$.
   $\\ \\$
   Let $g$ denote the y-coordinate of the lowest b-point in the $j^{th}$ clause pattern.
   There are a total $3m$ b-points of the form $b_{i,j}$ and $\overline{b}_{i,j}$ in all the variable patterns
   and these b-points are below the b-points of any clause pattern. 
   There are two b-points occupying two distinct horizontal lines in each clause pattern.
   So,   $g = a - 3m -2(j-1)$. 
   \item 
   The l-points remain in the horizontal lines already assigned to them. However, they may lie on 
   the c-line to satisfy the clause or another vertical line of the clause pattern.
   Assign coordinates 
     $( l^v_{f+3}, l^h_y)$ to $l_{i,j}$ (or, $\overline{l}_{i,j}$), 
   where 
   $l^h_y$ is the second
   component of coordinates assigned to
   $l_{i,j}$ (respectively, $\overline l_{i,j}$) earlier. 

    \item  
%
%
    Assign coordinates $(l^v_{f+10},l^h_{g})$ and $(l^v_{f+11},l^h_{g-1})$ to $b^c_{j,1}$ and  $b^c_{j,2}$ respectively.
    Whichever l-points of $C_j$ lie on there associated lines, 
    are blocked by $b^c_{j,1}$ and $b^c_{j,2}$ from $c_{j,1}$. 
   \item 
   The c-point of the clause pattern lies on the rightmost vertical line of the clause pattern and the bottommost
   horizontal line of the grid.
   So, assign coordinates  
   $( l^v_{f+3}  ,  l^h_{  \alpha }  )$
   to $c_i$.
   \end{enumerate}
Before we define the edge set of $G''$, we need the following definitions related to coordinates assigned to the 
vertices of $V''$.
For every vertex $q \in V'' \setminus\{v_1, v_2 \} $, let $S^q$ be the set of all pairs of coordinates assigned to $q$.
Furthermore, for every vertex $q \in V'' \setminus\{v_1, v_2 \} $, let $S^q_x$ and $S^q_y$ be the 
sets of the first and second components, respectively, of all pairs of coordinates 
assigned to $v$. 
$\\ \\$
Consider vertices $v_a$ and $v_b$, such that $v_a \in V'$ and $v_b \in V \setminus \{v_1, v_2 \}$.
Suppose that there exists some 
$l^v_{x_1} \in S^{v_a}_x \cap S^{v_b}_x$ such that
$(l^v_{x_1},l^h_{y_1}) \in S^{v_a}$ and $(l^v_{x_1},l^h_{y_2}) \in S^{v_b}$
for some $y_1$ and $y_2$. Then we refer to the pair $(v_a,v_b)$ as a \emph{vertical neighbouring pair} if 
there is no 
$v_c$ with
$v_c \neq v_a$ and $v_c \neq v_b$ and $(l^v_{x_1},l^h_{y_3}) \in S^{v_c}$ such that $y_1 > y_3 > y_2$. 
Similarly, suppose that 
there exists some 
$l^h_{y_1} \in S^{v_a}_y \cap S^{v_b}_y$ such that
$(l^v_{x_1},l^h_{y_1}) \in S^{v_a}$ and $(l^v_{x_2},l^h_{y_1}) \in S^{v_b}$
for some $x_1$ and $x_2$. Then we refer to the pair $(v_a,v_b)$ as a \emph{horizontal neighbouring pair} if 
there is no 
$v_c$ with
$v_c \neq v_a$ and $v_c \neq v_b$ and $(l^v_{x_3},l^h_{y_1}) \in S^{v_c}$ such that $x_1 < x_3 < x_2$. 
Let $L(G'')$ be the set of all such vertical or horizontal neighbouring pairs
possible from the vertices of $V'' \setminus \{v_1, v_2 \}$.
So, we have,
\begin{eqnarray*}
 E'' &=& E \cup E' \cup \{ (v_a,v_b) \mid v_a \in V' \ \textrm{and} \ v_b \in V \setminus \{v_1, v_2 \} \ \textrm{and} \  \\ 
     & & ( (S^{v_a}_x \cap S^{v_b}_x)\cup(S^{v_a}_y \cap S^{v_b}_y ) = \phi  \ \textrm{or} \  (v_a,v_b) \in L(G'')) \}\\            
\end{eqnarray*}
Based on the construction of $G''$, we state the following lemma without proof.
\begin{figure}
\begin{center}
\centerline{\hbox{\psfig{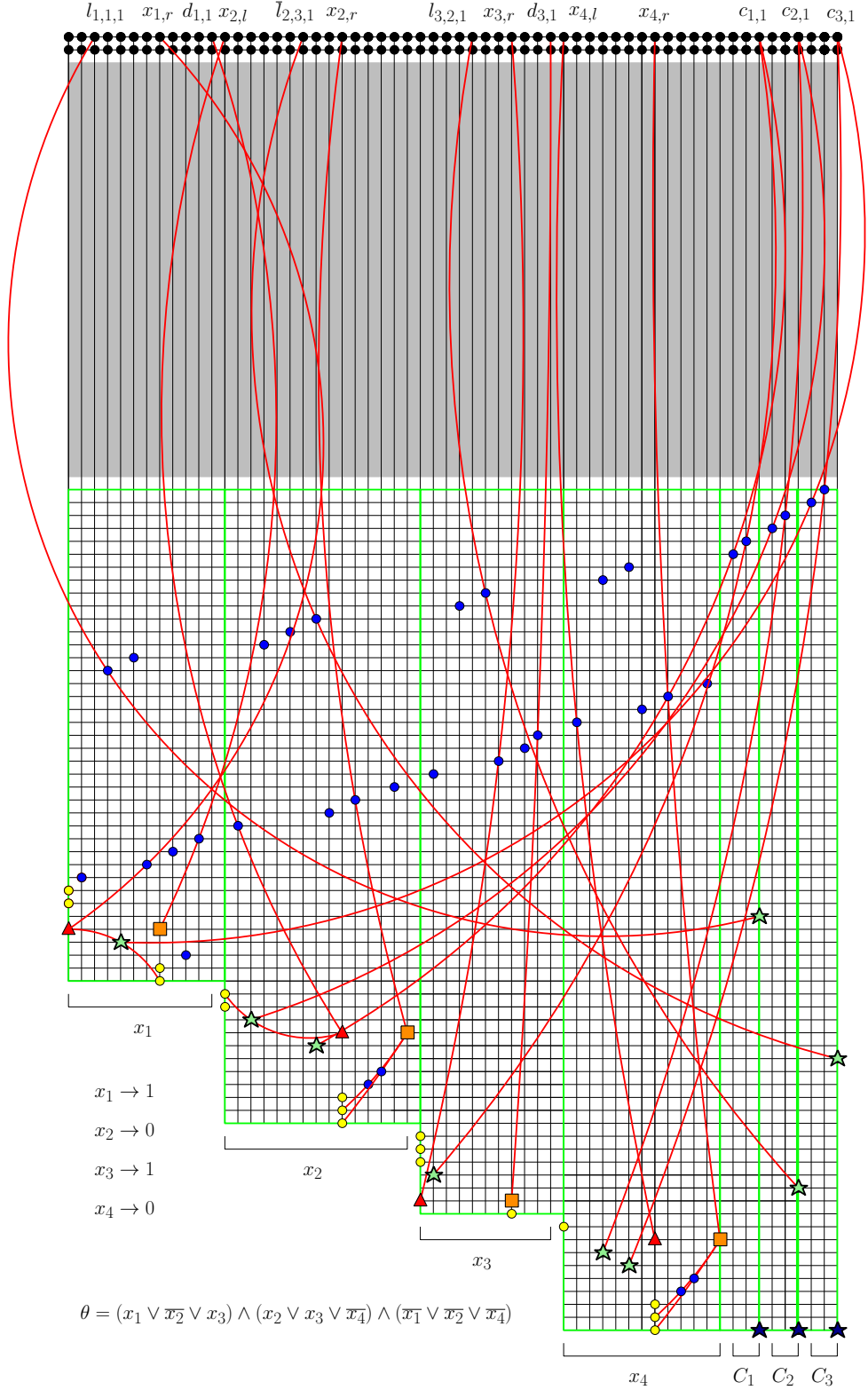}}}
\caption{
A canonical embedding $\psi$ of $G''$ corresponding to the given 
3-SAT formula. 
The top two rows of embedding points are $L_3$ and $L_4$.
The o-points, b-points, l-points, c-points, t-points and d-points
are 
depicted as pale circles, dark circles, pale stars, dark stars, triangles and squares respectively.
The non-horizontal and non-vertical line segments constitute $S_L$.
Some non-edges corresponding to a visibility embedding, drawn here as
red curves along with their blockers are as follows:
 $(t_1,x_{1,r}) \rightarrow   b_1^2            $,
 $(t_1,\overline o_{1,3})  \rightarrow    \overline l_{1,3}              $,
 $(d_1,d_{1,1})     \rightarrow   b_1^4            $,
 $(\overline l_{1,3}, c_{3,1})    \rightarrow    b^c_{3,1}     $,
 $(t_2, x_{2,l})    \rightarrow    b_2^1        $,
 $(t_2, o_{2,2})    \rightarrow    l_{2,2}       $,
 $(d_2,x_{2,r})    \rightarrow   b_2^3        $,
 $(d_2,\overline o_{2,1})      \rightarrow  b^d_{2,1}       $,
  $(d_2,\overline o_{2,3})    \rightarrow    b^d_{2,3}        $,
   $(  l_{2,2}, c_{2,1})    \rightarrow    b^c_{2,1}       $,
    $(\overline l_{2,1}, c_{1,1})    \rightarrow   b^c_{1,1}        $,
  $(t_3,x_{3,r})   \rightarrow    b_3^2         $,
   $(d_3,d_{3,1})   \rightarrow   b_3^4         $,
    $(  l_{3,1}, c_{1,1})    \rightarrow    b^c_{1,2}        $,
   $(t_4, x_{4,l})   \rightarrow   b_4^1         $,
   $(d_4,x_{4,r})   \rightarrow    b_4^3        $,
   $(d_4,\overline o_{4,2})    \rightarrow   b^d_{4,2}        $,
  $(d_4,\overline o_{4,3})    \rightarrow   b^d_{4,3}        $,
    $(\overline l_{4,2}, c_{2,1})     \rightarrow   b^c_{2,2}        $,
      $(\overline l_{4,3}, c_{3,1})    \rightarrow    b^c_{3,2}         $,
   $(l_{1,1},l_{1,1,1})    \rightarrow    b_{1,1}       $,
    $(l_{3,2},l_{3,2,1})    \rightarrow   b_{3,2}         $,
   $(\overline l_{2,3},\overline l_{2,3,1})    \rightarrow     \overline b_{2,3}      $.
}
\label{3SAT}
\end{center}
\end{figure}

\begin{lemma} \label{polytime}
 Given a 3-SAT formula $\theta$, the corresponding reduction graph $G''$ can be constructed in time polynomial 
in the size of $\theta$.
\end{lemma}
\subsection{Canonical embeddings of reduction graphs} \label{propgnew} 
As stated earlier, we have shown the construction of the reduction graph $G''$
of $\theta$ in polynomial time. 
We study here
some properties of $G''$. We need some definitions before we study these properties.
An embedding $\psi$ of $G''$ is called a 
\emph{canonical embedding} of $G''$ if (a) the embedding points of $\psi$ restricted to the vertices of $G$,
form the unique visibility embedding of $G$, and (b) for all $v_q \in G'$, the embedding point of $v_q$  
is embedded on the intersection of horizontal and vertical lines 
giving 
a pair of coordinates
that has been 
assigned to $v_q$.
Observe that in a canonical embedding, the following hold true.  
\begin{enumerate} [(i)]
 \item Each b-point is embedded only on its
corresponding b-line.
\item Each c-point is embedded only on its corresponding c-line. 
\item Each t-point is embedded only on
one of its two o-lines.
\item Each d-point is embedded only on either its d-line or its right o-line.
\item Each o-point is embedded only on its o-line.
\item   Each l-point is embedded either on its associated-line
or its c-line.
\end{enumerate}
%
If a canonical embedding $\psi$ of $G''$ is also a visibility embedding of $G''$, 
then $\psi$ is called a \emph{canonical visibility embedding} of $G''$.
We have the following lemma.
\begin{lemma} \label{canemb}
 If $G''$ is a PVG then every visibility embedding of $G''$ is a canonical visibility embedding.
\end{lemma}
\begin{proof}
We know from Lemmas \ref{4lp} and \ref{uemsgg} that $G$ has a unique visibility embedding. 
Let $\xi$ be the unique visibility embedding of $G$. Consider lines $L_1$, $L_2$, $L_3$ and $L_4$ in $\xi$ as before (Figure \ref{msgg}).
Note that $\vert L_3 \vert = \vert L_4 \vert = 2(\alpha)^4$
Let $G''$ be a PVG and $\xi'$ be a visibility embedding of $G''$.
Observe that the total number of embedding points in $\xi' \setminus (L_1 \cup L_2 \cup L_2 \cup L_4)$
is less than $\alpha$. Moreover, the embedding points corresponding to the vertices of $G'$ 
are visible from most embedding points of $L_1$, $L_2$, $L_3$ and $L_4$. So, $G''$
satisfies the conditions of 
 Lemmas \ref{4lp} and \ref{uemsgg}, and by a similar argument, it can be shown that 
 the embedding points of $\xi'$ restricted to the vertices of $G$,
form the unique visibility embedding of $G$.
$\\ \\$
Now we show that every vertex $v_q \in G'$ satisfies the second condition of a canonical embedding.
Consider the embedding point $l^v_i \cap L_3$ in $\xi'$. Its corresponding vertex, by the construction of $G''$,
is not adjacent to $v_q$ if and only if $l^v_i$ is assigned as a coordinate to $v_q$.
A similar argument follows for $v_q$ and embedding points of the form $l^h_j \cap L_1$ in $\xi'$.
On the other hand,
two non-consecutive embedding points on a horizontal or vertical line cannot be visible from each other.
So, 
the embedding point of $v_q$  
is embedded on the intersection of horizontal and vertical lines 
giving 
a pair of coordinates
that has been 
assigned to $v_q$. Hence, $\xi'$ is a canonical visibility embedding of $G''$.
\end{proof}
\noindent 
Let us define the {variable pattern} of each $x_i$ and the {clause pattern} of each $C_j$.
For each $x_i$, let
$a = \alpha -7n - 6m - \sum\limits_{x_i \in \theta}\overline{n_i}) 
  -1$,
  $b =   2(\alpha)^4 + \sum _{k=1} ^{i-1} (2n_{k} + 3\overline{n_{{k}}} +7)$,
  and $c=  \sum _{k=1} ^{i-1} (2n_{k} + 3\overline{n_{{k}}} + 3)$, as defined in Section \ref{subsecconsred}.
%
%
%
%
%
%
%
For a canonical embedding $\xi$ of $G''$, the closed region bounded by the four lines 
 $l^v_{b+1}$,
 $l^v_{b+ 2n_{i} + 3\overline{n_{i}} +7}$,
  $L_3$
  and 
$l^h_{a+ 4n+c + 2n_{i} +2\overline{n_{i}} + 3 + \overline{n_i}}$
is called the \emph{variable pattern} of $x_i$ (Figure \ref{varpat}).
Let, for each $C_j$, let
$f = 2(\alpha)^4 + (\sum _{k=1} ^{n} (2n_{k} + 3n_{\overline{k}} +7)) + 12(j-1)$, as defined in Section \ref{subsecconsred}.
For a canonical embedding $\xi$ of $G''$,
the closed region bounded by the four lines 
$l^v_f$,
 $l^v_{f+12}$,
$L_3$
and 
$l^h_{\alpha}$
is called the \emph{clause pattern} of $C_j$ (Figure \ref{clpat}).
\begin{lemma} \label{satcan}
 If $\theta$ is not satisfiable, then $G''$ does not have a canonical visibility embedding.
\end{lemma}
\begin{proof}
 Assume on the contrary that $\theta$ is not satisfiable but $G''$ has a canonical visibility embedding $\xi'$. 
 So, each t-point of $\xi'$ is embedded on either its left o-line or right o-line. So, the embedding of the t-points
 corresponds to an assignment of the variables of $\theta$.
 Since one of the clauses (say, $C_j$) is not satisfied, 
 the complements of the literals in $C_j$ have been assigned to $1$. Hence, if $l_{i,j} \in V^c_j$ 
 then $t_i$ lies on the left o-line of $x_i$ and $l_{i,j}$ must be embedded in the variable pattern of $x_i$ in $\xi'$.
 A similar argument holds if $\overline l_{i,j}$ is in $V^c_j$. This is true for all three literals of $C_j$. Hence,
 no l-point can be embedded in the clause pattern of $C_j$ in $\xi'$. Therefore, there is no embedding point 
 to block the visibility of the c-point from  $c_{j,2}$,
 contradicting the assumption.
\end{proof}
\begin{lemma} \label{redfrstdir}
 If $\theta$ is not satisfiable, then $G''$ is not a PVG.
\end{lemma}
\begin{proof}
 The proof follows from Lemmas \ref{canemb} and \ref{satcan}.
\end{proof}

\subsection{Reduction from 3-SAT} \label{3satred}
In this Section we prove that if $\theta$ is satisfiable then $G''$ is a PVG.
Recall that if $\theta$ is not satisfiable then $G''$ is not a PVG.
 We start by constructing a canonical embedding $\psi$ of $G''$,
 and then transform it into a canonical visibility embedding of $G''$.
Let $S_\theta$ be a satisfying assignment of $\theta$.
 Since $G'' = G \cup G'$, 
  all the embedding points corresponding to the vertices of $G$ are embedded initially to form the unique visibility embedding of
  $G$. 
  Then, embedding points corresponding to $G'$ are embedded to complete the embedding $\psi$ of $G''$ (Figure \ref{3SAT})
  as follows.
  Repeat the following three steps for all $x_i \in \theta$.
   \begin{enumerate} [(a)]
  \item If $x_i$ is assigned $1$ in $S_\theta$ then embed the t-point of $t_i$ on its left o-line. 
  Otherwise embed the t-point of $t_i$ on its right o-line.
  \item If the t-point of $t_i$ is embedded on the its left o-line, then embed the d-point of $d_i$ on its right o-line.
  Otherwise embed the d-point of $d_i$ on its d-line.
  \item If the t-point of $t_i$ is embedded on its left o-line, then embed the l-points of $l_{i,j}$
  on their associated-lines, for all $j$.
  Otherwise embed the l-points of $\overline l_{i,j}$
  on their associated-lines, for all $j$.
 \end{enumerate}

  \noindent 
As a next step,
  for each clause $C_j$, choose an l-point of $C_j$ that has 
  not been embedded yet, and embed it on the intersection of the c-line
  and a horizontal line corresponding to a pair of coordinates assigned to its c-vertex.
  Observe that such l-points are always available for each clause, since $S_\theta$ is a satisfying assignment
  of $\theta$. All the remaining l-points are embedded on their associated-lines.
 The construction of $\psi$ is completed by the following step.
  \begin{enumerate} [(a)]
  \item Embed all the c-points and b-points on the intersection points representing the unique pair of coordinates assigned to them.
 \end{enumerate}
Before the above embedding $\psi$ is transformed to a visibility embedding $\xi$ of $G''$,
we need the following lemma for rotating a line in $\psi$.
 \begin{lemma}
 \label{move1}
 Consider a line $l'$ 
 of $\psi$. Let $\{ p_1, p_2, \ldots, p_q\}$ denote the order of all embedding points on $l'$
 where $p_i$ 
 lies on the intersection point of $l'$ and a non-ordinary line $l_i$  (Figure \ref{befig1}(a)).
For any given real $\epsilon > 0$ and 
embedding point $p_j$ for $1 \leq j \leq q$, $l'$ can be rotated with $p_j$ 
as the pivot to form a new $l'$ satisfying the following properties.
%
 \begin{enumerate}[(a)]
 \item \label{l1pa}
 The 
 embedding points of $\psi$ on $l'$, except $p_j$, are relocated on the new $l'$. All other embedding points in $\psi$ remain unchanged.

\item \label{l1pd}
The order of embedding points on $l'$ and the new $l'$ are the same.
\item  \label{l1pc}
The order of embedding points on each $l_i$ also remains the same.
  \item  \label{l1pe}
  $\forall i \neq j, 1 \leq i \leq q$, $p_i$ does not lie on any other non-ordinary line.
  \item \label{l1pf}
  For each $p_i$ on $l'$, the Euclidean distance between the new and old positions of $p_i$ is less than  $\epsilon$.
 \end{enumerate}

\end{lemma}
\begin{proof} Rotate $l'$ with $p_j$ as the pivot in clockwise direction
 until it reaches a point $y$ on some line $l_i$ such that $y$ is either an intersection point of $\psi$ or 
 the length of the segment $p_iy$ is $\epsilon$.
%
%
%
%
The new $l'$ is the line through 
 $p_j$ and 
 some point in the interior of $p_iy$.
 Embed each $p_i$ 
 on the intersection point of $l_i$ and the new $l'$ (Figure \ref{befig1}(b)).
 It can be seen that the properties (\ref{l1pa}), (\ref{l1pc}), (\ref{l1pd}), (\ref{l1pe}) and (\ref{l1pf})
 of the lemma are satisfied.
\end{proof}
\noindent 
Observe that in $\psi$, there can be several non-ordinary lines that are not horizontal or vertical lines.
The blocking relationships induced by these lines may not conform to the edges in $G'$.
Treating
each vertical line as $l'$ and each horizontal line
intersecting
$l'$ as $l_i$,
 Lemma \ref{move1} 
 is applied 
 on every vertical line in $\psi$ by rotating around $p_2$.
  Thus, any non-ordinary line that now passes through an embedding point of $\psi$ is either a vertical or a horizontal line. 
We have the following lemma on rotating multiple lines of $\psi$.
\begin{figure}
\begin{center} 
\centerline{\hbox{\psfig{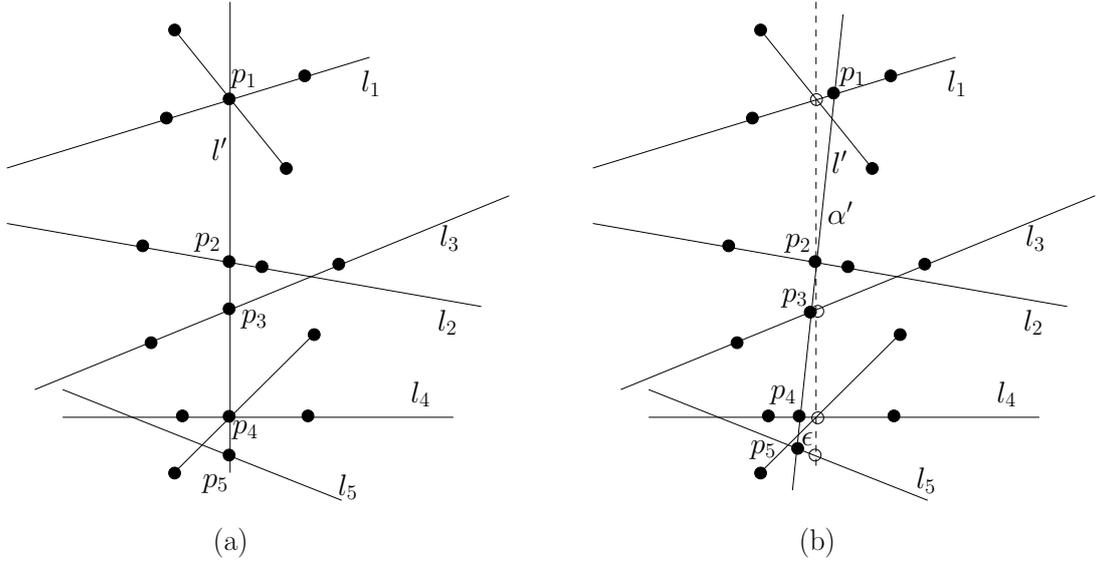}}}
\caption{
(a) The lines $l_1, l_2, l_3, l_4$ and $l_5$ intersect $l'$
at $p_1$, $p_2$, $p_3$, $p_4$ and $p_5$ respectively. 
(b) The line $l'$ is rotated around $p_2$, 
so that all embedding points on $l'$
except $v_2$ are relocated. Each of the relocated embedding points lie on exactly two non-ordinary lines and within
$\epsilon$ distance of their previous positions.
}
\label{befig1}
\end{center}
\end{figure}
\begin{lemma} 
\label{move3}
  Consider a vertical line $l'$ 
 of $\psi$. 
 Let $( w_i, w_{i+1}, \ldots, w_{i+j-1}, w_{i+j} )$ be all embedding points on $l'$ from $w_i$ to $w_{i+j}$
 such that they lie on the intersection points of $l'$ with $( l_i, l_{i+1}, \ldots, l_{i+j-1}, l_{i+j})$ respectively.
%
%
%
 Let $q_1$ and $q_2$ be any two designated points on the interval
 $w_{i+1}w_{i+j-1}$. 
 For every line $l_{i+k} \in ( l_{i+1}, \ldots, l_{i+j-1})$, a new $l_{i+k}$ can be constructed such that 
 $l_{i+k}$ intersects $l'$ at a point $r_{k}$ satisfying the following properties.
  \begin{enumerate}[(a)]
  \item \label{l2pa}
  The points $( r_1, r_2, \ldots, r_{j-1} )$ lie
  on $q_1q_2$ and their order follows the order of 
  $( w_i, w_{i+1}, \ldots, w_{i+j-1}, w_{i+j} )$.
%
   \item  \label{l2pb}
   The non-ordinary lines passing through the embedding points on 
   $( l_i, l_{i+1}, \ldots, l_{i+j-1}, l_{i+j} )$ are either vertical or horizontal lines.
%
  \end{enumerate}

\end{lemma}
\begin{proof} 
Let $p'_1$ be a point on $q_1q_2$. 
Set $\epsilon = \frac{1}{2} \min (q_1r'_1, r'_1q_2)$
Rotate the line passing through $r'_1$ and $p_1$ with 
$p_1$ as the pivot
using Lemma \ref{move1} to obtain a new intersection point $r_1$ on $l'$. The line passing through $p_1$ and 
$r_1$ is the new $l_1$, and embedding points on $l_1$ are relocated on the corresponding intersection points of 
the new $l_1$.
Analogously, choose a point $r'_2$ on $r_1q_2$ and construct the new $l_2$ giving a new intersection point 
$r_2$ of $l'$ on $r_1q_2$.
These operations are performed on all lines in $( l_{i+1}, \ldots, l_{i+j-1})$. 
It can be seen that the properties (\ref{l2pa}) and (\ref{l2pb})
 of the lemma are satisfied.
\end{proof}
\begin{figure}
\begin{center} 
\centerline{\hbox{\psfig{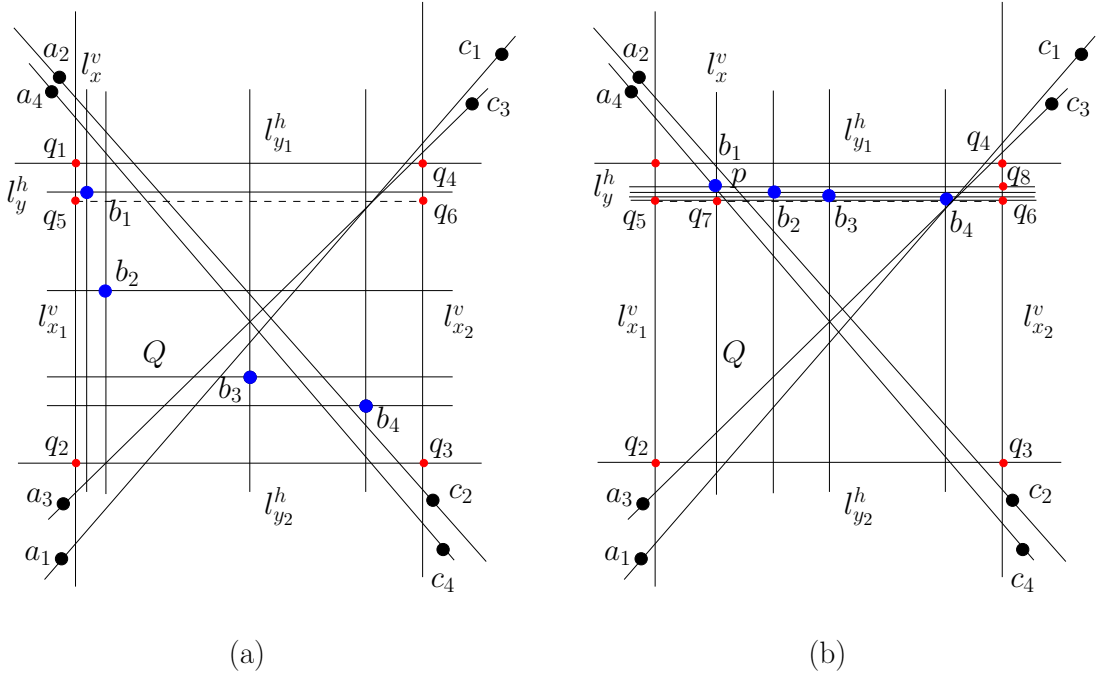}}}
\caption{
(a) The embedding points $b_1$, $b_2$, $b_3$ and $b_4$ of $B$ lie in the interior of the quadrilateral $q_1q_2q_3q_4$.
(b) The embedding points $b_1$, $b_2$, $b_{3}$ and $b_4$ of $B$ are relocated inside $q_1q_5q_6q_4$,
with $b_1$ blocking the segment $a_4c_4$. 
}
\label{befig2}
\end{center}
\end{figure}

$\\ \\$
Using Lemma \ref{move3}, we show that embedding points inside a special type of quadrilateral can be relocated 
as blockers of pairs of embedding points lying outside the quadrilateral.
Consider a quadrilateral $Q = (q_1, q_2, q_3, q_4)$, where $q_1$, $q_2$, $q_3$ and $q_4$
are embedding points of $\psi$ lying on 
$(l^v_{x_1},l^h_{y_1})$,
$(l^v_{x_1},l^h_{y_2})$, $(l^v_{x_2},l^h_{y_2})$ and $(l^v_{x_2},l^h_{y_1})$ respectively,
and $x_1<x_2$ and $y_1<y_2$.
Let $B$
be the set of all embedding points lying inside $Q$.
$B$ is said to be an \emph{ordered set} if no two embedding points of $B$ lie on the same horizontal or 
vertical line, and $B$ satisfies exactly one of the following properties.
\begin{enumerate}
 \item For all embedding points $b' \in B$ and $b'' \in B$ embedded on $(l^v_{x'},l^h_{y'})$ and $(l^v_{x''},l^h_{y''})$ respectively,
  if $x'<x''$ (or, $x'>x''$) then $y'<y''$ (respectively, $y'>y''$).
   \item For all embedding points $b' \in B$ and $b'' \in B$ embedded on $(l^v_{x'},l^h_{y'})$ and $(l^v_{x''},l^h_{y''})$ respectively,
  if $x'<x''$ (or, $x'>x''$) then $y'>y''$ (respectively, $y'<y''$).

\end{enumerate}

$\\ $
Let $A$
be a set of embedding points of $\psi$ such that 
each $a_i \in A$ lies to the left of $l^v_{x_1}$ and also lies either above $l^h_{y_1}$ or below $l^h_{y_2}$.
Let $C$ 
be a set of embedding points of $\psi$ such that 
each $c_i \in C$ lies to the right of $l^v_{x_2}$ and  also lies either above $l^h_{y_1}$ or below $l^h_{y_2}$.
Let $S$
be a set of line segments $a_ic_j$ where $a_i \in A$ and $c_j \in C$,
and $a_ic_j$ intersects both $q_1q_4$ and $q_2q_3$. 
A pentuple $U = (Q,A,B,C,S)$ is called a \emph{good pentuple} if 
$ \vert B \vert \geq \vert S \vert $, and
   $B$ is an ordered set.
\begin{lemma} 
 \label{mainmovelem1}
 For a given good pentuple $U = (Q,A,B,C,S)$ in $\psi$, horizontal and vertical lines passing through 
 the embedding points of $B$ can be relocated satisfying the following properties.
  \begin{enumerate}[(a)]
  \item All horizontal and vertical lines in $\psi$ retain their angular ordering around $p_1$ and $p_2$ respectively.
   \item Each embedding point in $B$ lies on exactly one segment of $S$.
   \item Each embedding point in $B$ lies on exactly three non-ordinary lines, two of which are horizontal and vertical lines.
   \item 
   For every horizontal or vertical line $l''$ containing $b \in B$, no embedding point on 
   $l'' \setminus \{ b \}$ lies on a third non-ordinary line after relocation.
  \end{enumerate}
\end{lemma}
\begin{proof}
Wlog let $B$ satisfy Property $1$ of ordered sets. Choose an appropriate point $q_5 \in q_1q_2$ 
such that no intersecting points of the segments of $S$ lie
in the interior of $q_1q_5q_6q_4$, where $q_6$ is the point of intersection of $p_1q_5$ and $q_3q_4$
(Figure \ref{befig2}(a)).
Let $H_Q$ and $V_Q$ be the set of all horizontal and vertical lines passing through $Q$, respectively.
By applying Lemma \ref{move3} on any vertical line in $V_Q$, relocate all horizontal lines
of $H_Q$ such that all of them pass through $q_1q_5q_6q_4$.
$\\ \\$
Observe that all embedding points of $B$ have moved inside $q_1q_5q_6q_4$.
Since none of the segments of 
$S$ intersect inside $q_1q_5q_6q_4$, they have a left to right order defined by their intercepts on $q_5q_6$.
Let $a_ic_j$ be the leftmost segment of $S$ in this order.
Denote the leftmost embedding point of $B$ as $b$, and let $l^v_x$ and $l^h_y$ be its vertical and horizontal lines 
respectively.
Applying Lemma \ref{move3} on any horizontal line in $H_Q$, all vertical lines 
of $V_Q$ are relocated such that $l^v_x$ intersects $a_ic_j$ at a point (say, $p$) (Figure \ref{befig2}(b)),
maintaining other lines of $V_Q$ passing through $q_1q_5q_6q_4$.
Treating $p$ as an embedding point and taking $p_2$ as the pivot, Lemma \ref{move1} can be applied on
$l^v_x$ to ensure that $p$ does not lie on any other non-ordinary line.
Now embed $b$ on $p$ by relocating $l^h_y$ accordingly.
Relocate all other horizontal lines
of $H_Q$ by applying Lemma \ref{move3}, maintaining all lines of $H_Q$ passing through $q_1q_5q_6q_4$.
$\\ \\$
It can be seen that $U' = (bq_7q_3q_8,A,B \setminus \{ b \},C,S \setminus \{ a_ic_j \} )$ is a good pentuple, 
where $q_7 = l^v_x \cap q_5q_6$
 and $q_8 = l^h_y \cap q_3q_4$.
 Repeating the above procedure, 
 embedding points of $B$
 are placed on all segments of $S$ as blockers,
 satisfying  properties (a), (b), (c) and (d) of the lemma.
 Analogous arguments of the proof are applicable if $B$ satisfies Property $2$ of ordered sets.
\end{proof}
$\\$
Now we use Lemmas \ref{move1}, \ref{move3} and \ref{mainmovelem1} to finally transform $\psi$ into a 
visibility embedding $\xi$ of $G''$. We have the following lemma.
\begin{figure}
\begin{center} 
\centerline{\hbox{\psfig{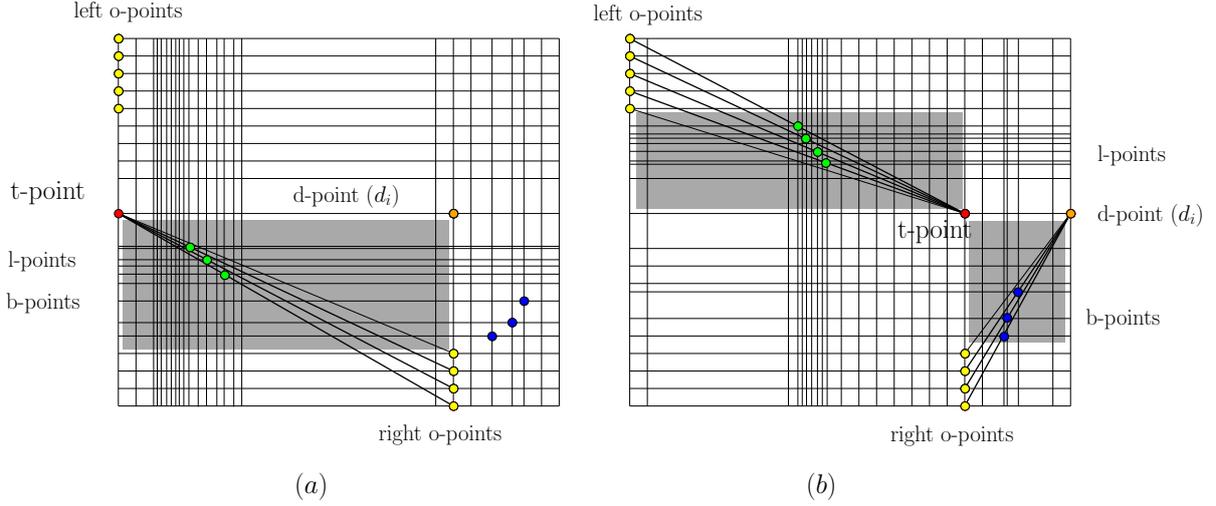}}}
\caption{(a)
Since t-point lies on the left o-line, blockers are placed between t-point and right o-points.
The quadrilateral required for a good pentuple is 
shaded in gray.
(b)
Since t-point lies on the right o-line, blockers are placed between t-point and left o-points.
Since b-point lies on the d-line, blockers are placed between d-point and right o-points.
Quadrilaterals required for a good pentuple are
shaded in gray.
}
\label{embex}
\end{center}
\end{figure}
\begin{lemma} \label{mainmovelem2}
 The canonical embedding $\psi$ can be transformed into a visibility embedding  $\xi$ of $G''$.
\end{lemma}
\begin{proof}
 The only adjacency relationships of $G''$ that $\psi$ may not satisfy are those $(i)$ between o-vertices and t-vertices,
 $(ii)$ between o-vertices and d-vertices, and $(iii)$ between t-vertices, l-vertices, 
 d-vertices and vertices corresponding to certain points
 on $L_3$.
Consider $(i)$ and $(ii)$.
 For each $x_i$, if the t-point of $t_i$ is embedded on its left o-line, then consider the quadrilateral $Q$ 
 formed by the horizontal line passing through the topmost right o-point, horizontal line of the t-point,
 left o-line and right o-line. Draw two more vertical and horizontal lines such that they form a quadrilateral $Q'$ 
 in the interior of $Q$, and only nominally smaller than $Q$. It can be seen that $Q'$, $t_i$, the right o-points,
 the l-points of the form $\overline l_{i,j}$, and all segments between the t-point and the right o-points
 form a good pentuple. Hence these segments can be blocked by relocating the corresponding l-points using Lemma 
 \ref{mainmovelem1} (Figure \ref{embex} (a)).
 A similar argument works if the t-point is embedded on the right o-line (Figure \ref{embex} (b)), or if the d-point is embedded
 on the d-line (Figure \ref{embex} (b)).
 $\\ \\$
 Consider $(iii)$.
 Let $S_L$ be all such segments having an endpoint on $L_3$   (Figure \ref{3SAT}).
 Locate a point $p_3$ on $l^h_1$ such that 
 $L_3$ and $L_4$ are above and below $p_1p_3$ respectively.
Moreover, the intersection points of $S_L$ with vertical lines of $\psi \setminus \{L_3 \}$ lie below $p_1p_3$.
 Let $H_L$ be the set of all horizontal lines between $L_3$ and $l^h_{a+4n+1}$,
 where $a = \alpha -7n - 6m \sum\limits_{x_i \in \theta}\overline{n_i} - 1$  as stated in Section \ref{subsecconsred}. 
 Apply lemma \ref{move3} on any vertical line of $\psi$, and treating its intersection points with horizontal lines as embedding points, 
 all horizontal lines of $H_L$ are relocated
 so that they are above 
 ${p_1p_3}$.
  Consider any segment $s_j \in S_L$. Let the two endpoints of $s_j$ in $\psi$ be $r_1(s_j)$ and $r_2(s_j)$, where $r_1(s_j) \in L_3$. Let the 
  two vertical lines passing through $r_1(s_j)$ and $r_2(s_j)$ be $l^v_{u_j}$ and $l^v_{w_j}$ respectively. Observe that if 
  $u_j < w_j$ (or, $u_j > w_j$) then $l^v_{u_j+1}$ (respectively, $l^v_{u_j-1}$) contains a b-point lying on a horizontal line of $H_L$,
  due to the construction of $\psi$.
 Such a b-point exists for every segment in $S_L$.
 For two segments of $S_L$ with a common endpoint on a c-line, the two b-points on the two vertical lines immediately to the left of the c-line 
 correspond to the two segments.
 Let $B_L$ denote the set of all these b-points.
Now consider a b-point $b_i \in B_L$ such that the horizontal line passing through $b_i$ (say, $l^h_{y_i}$) is lower than the 
horizontal line passing through any other b-point of $B_L$.
$ \\ \\$
Let $s_i \in S_L$ be the segment corresponding to $b_i$. 
Let $Q_i$ be the quadrilateral enclosed by $l^v_{u_i}$, $l^v_{u_{i+2}}$, ${p_1p_3}$
and $l^h_{y_i -1}$, assuming $u_i < w_i$. 
Observe that $Q_i$, $B_i = \{ b_i \}$, $A_i = \{r_1(s_i) \}$, $C_i = \{r_2(s_i) \}$ and $S_i = \{ s_i \}$
constitute a good pentuple, say, $U_i$. Apply Lemma \ref{mainmovelem1} on $U_i$ to place $b_i$ as a blocker on $s_i$.
Remove $b_i$ and $s_i$ from $B_L$ and $S_L$ respectively. Remove $l^h_{y_i}$ and all horizontal lines below it from $H_L$.
Repeat the process on the lowest b-point of $B_L$, 
treating $l^h_{y_i}$ as the new $p_1p_3$.
$\\ \\$
It may so happen that the same embedding point on $L_3$ is the endpoint of two segments $s_i$ and $s_j$ in $S_L$,
i.e., $r_1(s_i) = r_1(s_j)$.
This case arises only when 
$r_1(s_i)$ and $r_1(s_j)$
lie
on a c-line of $\psi$.
In this case, the two b-points on the two vertical lines immediately to the left of the c-line are relocated
as blockers on $s_i$ and $s_j$, using an analogous process.
$\\ \\$
Hence, b-points can be assigned as blockers on segments of $S_L$ in cases $(i)$, $(ii)$ and $(iii)$. Therefore,
the canonical embedding $\psi$ can be transformed into a visibility embedding  $\xi$ of $G''$.
\end{proof}
$\\$
Finally, we have the following theorem.
\begin{theorem}
 The recognition problem for PVGs in NP-hard.
\end{theorem}
\begin{proof}
 Given a 3-SAT formula $\theta$, the construction of the corresponding graph $G''$ takes polynomial time, due to Lemma \ref{polytime}.
 The graph $G''$ is a PVG if and only if $\theta$ is satisfiable, due to Lemma \ref{redfrstdir} and Lemma \ref{mainmovelem2}.
 Hence the recognition problem for PVGs in NP-hard.
\end{proof}
\begin{cor}
 The reconstruction problem for PVGs in NP-hard.
\end{cor}


\section{Concluding remarks} \label{conclrem}
In this paper we have proved that
the recognition and reconstruction problems for point visibility graphs, are
NP-hard. 
On the other hand, we know that the recognition problem for point visibility graphs is in PSPACE \cite{recogpvg-2014}.
It has been pointed out by Ghosh and Goswami \cite{prob-ghosh} that the recognition problem for point visibility graphs,
and to show whether the problem lies in NP, are still open. 


\section{Acknowledgements} 
The author would like to thank Jean-Lou De Carufel,
Anil Maheshwari and Michiel Smid for the many discussions that 
helped to structure the paper.
The author would also like to thank Amitava Bhattacharya,  
and Prahladh Harsha for their valuable suggestions. 
The author is specially thankful to the anonymous referees and Subir Kumar Ghosh 
 for their suggestions which have improved the presentation of the paper significantly.
\bibliographystyle{plain}
\bibliography{vis}

\begin{thebibliography}{10}

\bibitem{cgl-pgd-85}
B.~Chazelle, L.~J. Guibas, and D.T. Lee.
\newblock The power of geometric duality.
\newblock {\em BIT}, 25:76--90, 1985.

\bibitem{bcko-cgaa-08}
M.~de\ Berg, O.~Cheong, M.~Kreveld, and M.~Overmars.
\newblock {\em Computational Geometry, Algorithms and Applications}.
\newblock Springer-Verlag, 3rd edition, 2008.

\bibitem{diestel}
R.~Diestel.
\newblock {\em Graph Theory}.
\newblock Springer-Verlag, 2005.

\bibitem{Edelsbrunner:1986:CAL}
H.~Edelsbrunner, J.~O'Rourke, and R.~Seidel.
\newblock Constructing arrangements of lines and hyperplanes with applications.
\newblock {\em SIAM Journal on Computing}, 15:341--363, 1986.

\bibitem{g-vap-07}
S.~K. Ghosh.
\newblock {\em Visibility Algorithms in the Plane}.
\newblock Cambridge University Press, 2007.

\bibitem{prob-ghosh}
S.~K. Ghosh and P.~P. Goswami.
\newblock Unsolved problems in visibility graphs of points, segments and
  polygons.
\newblock {\em ACM Computing Surveys}, 46(2):22:1--22:29, December, 2013.

\bibitem{recogpvg-2014}
S.~K. Ghosh and B.~Roy.
\newblock Some results on point visibility graphs.
\newblock In {\em Proceedings of the Eighth International Workshop on
  Algorithms and Computation}, volume 8344 of {\em Lecture Notes in Computer
  Science}, pages 163--175. Springer-Verlag, 2014.

\bibitem{kpw-ocnv-2005}
J.~K\'ara, A.~P{\'o}r, and D.~R. Wood.
\newblock {On the Chromatic Number of the Visibility Graph of a Set of Points
  in the Plane}.
\newblock {\em Discrete \& Computational Geometry}, 34(3):497--506, 2005.

\bibitem{lw-apcf-79}
T.~Lozano-Perez and M.~A. Wesley.
\newblock An algorithm for planning collision-free paths among polyhedral
  obstacles.
\newblock {\em Communications of ACM}, 22:560--570, 1979.

\bibitem{viscon-wood-2012}
M.~S. Payne, A.~P{\'o}r, P.~Valtr, and D.~R. Wood.
\newblock On the connectivity of visibility graphs.
\newblock {\em Discrete \& Computational Geometry}, 48(3):669--681, 2012.

\bibitem{p-vgps-2008}
F.~Pfender.
\newblock Visibility graphs of point sets in the plane.
\newblock {\em Discrete \& Computational Geometry}, 39(1):455--459, 2008.

\bibitem{sh-ddsg-79}
L.G. Shapiro and R.M. Haralick.
\newblock Decomposition of two-dimensional shape by graph-theoretic clustering.
\newblock {\em IEEE Transactions on Pattern Analysis and Machine Intelligence},
  PAMI-1:10--19, 1979.

\end{thebibliography}
\end{document}